\documentclass[preprintnumbers,amsmath,amssymb, amsart, twocolumn,pra,showpacs]{revtex4-1}

\usepackage{graphicx}
\usepackage{amsmath}
\usepackage{fullpage}
\usepackage{amsfonts}
\usepackage{amssymb}
\usepackage{dsfont}
\usepackage{amsthm}
\usepackage{tensor}
\usepackage{amsthm}
\usepackage{xcolor}
\usepackage{comment}
\usepackage{mathtools}
\usepackage{physics}
\usepackage{bbold}
\newtheorem{lemma}{Lemma}
\newtheorem{theorem}{Theorem}
\usepackage{appendix}
\usepackage[ruled,vlined]{algorithm2e}

\def\cO{\mathcal{O}}

\def\bl{{\bf \Lambda}}
\def\one{{\mathchoice {\rm 1\mskip-4mu l} {\rm 1\mskip-4mu l} {\rm
1\mskip-4.5mu l} {\rm 1\mskip-5mu l}}}
\def\dqc1{\textsc{DQC1}}

\makeatletter
\newtheorem*{rep@theorem}{\rep@title}
\newcommand{\newreptheorem}[2]{%
\newenvironment{rep#1}[1]{%
 \def\rep@title{#2 \ref{##1}}%
 \begin{rep@theorem}}%
 {\end{rep@theorem}}}
\makeatother

\newreptheorem{theorem}{Theorem}
\newreptheorem{lemma}{Lemma}

\newreptheorem{claim}{Claim}

\newreptheorem{definition}{Definition}

\newreptheorem{remark}{Remark}
\newtheorem{corollary}{Corollary}
\newreptheorem{corollary}{Corollary}

\newreptheorem{observation}{Observation}
\newtheorem*{remark*}{Remark}
\begin{document}
\title{Hamiltonian simulation in the low energy subspace}

\author{Burak \c{S}ahino\u{g}lu}\email{sahinoglu@lanl.gov}
\address{Theoretical Division, Los Alamos National Laboratory, Los Alamos, NM 87545, USA}

\author{
Rolando D. Somma}\email{somma@lanl.gov}
\affiliation{Theoretical Division, Los Alamos National Laboratory, Los Alamos, NM 87545, USA}

\date{\today}

\begin{abstract}
We study the problem of simulating the dynamics of spin systems when the initial state is supported on a subspace of low-energy of a Hamiltonian $H$. 
This is a central problem in physics with vast applications in many-body systems and beyond, where the interesting physics takes place in the low-energy sector.
We analyze error bounds induced by product formulas that approximate the evolution operator and show that these bounds depend on an effective low-energy norm of $H$. 
We find improvements over the best previous complexities of product formulas that apply to the general case, and these improvements are more significant
for long evolution times that scale with the system size and/or small approximation errors.
To obtain these improvements, we prove 
exponentially-decaying upper bounds on the leakage to high-energy subspaces due to the product formula.
Our results provide a path to a systematic study of Hamiltonian simulation at low energies, which will be required to push quantum simulation closer to reality.
\end{abstract}

\maketitle

\section{Introduction}
The simulation of quantum systems is believed
to be one of the most important applications of quantum computers~\cite{Fey82}. 
Many quantum algorithms for simulating quantum dynamics exist~\cite{Llo96,AT03,BAC07,WBH+10,BCC+14,BCC+15,LC17,LC19,Cam19,haah2018quantum}, with applications in physics~\cite{SOGKL02,JLP12}, quantum chemistry~\cite{WBCH14,BBK+15,BBK++15}, and beyond~\cite{CKS17}. While these algorithms are deemed efficient and run in time polynomial in factors such as system size, ongoing work has significantly improved the performance of such approaches. 
These improvements are important to explore the power of quantum computers and push quantum simulation closer to reality.

Leading {\em{Hamiltonian simulation}} methods are based on a handful of techniques.
A main example is the  product formula, which approximates the
evolution of a Hamiltonian $H$ by short-time evolutions under the terms that compose $H$~\cite{Suz90,Suz91,BAC07,WBH+10}. 
Each such evolution can be decomposed as a sequence of two-qubit gates~\cite{SOGKL02} to build up a quantum algorithm.
Product formulas are attractive for various reasons: 
they are simple, intuitive, and their implementations may not require ancillary qubits, which contrasts other sophisticated methods as those in Refs.~\cite{BCC+15,LC17}. 
Product formulas are also the basis of classical simulation algorithms
including path-integral Monte Carlo~\cite{NB98}.

Recent works provide refined error bounds of product formulas~\cite{Som16,CS19,CST+19,CBC20}. These works regard various settings, such as when $H$ is a sum of spatially-local terms or when these terms satisfy Lie-algebraic properties. 
Nevertheless, while these improvements are important and necessary, a number of shortcomings remain. 
For example, the best-known complexities of product formulas scale poorly with the norm of $H$
or its terms, which can be very large or unbounded,
even when the evolved quantum system does not explore {\em high-energy} states. 
These complexities may be improved under physically-relevant assumptions on energy scales.
In fact, numerical simulations of few spin systems suggest that product formulas applied to low-energy states lead to much lower errors than that of the worst-case. Figure~1, for example, 
shows these errors for a 2x6 spin-1/2 Heisenberg model,
suggesting that a complexity improvement is possible under a low-energy assumption on the initial state. Simulation results for related models present similar features. Nevertheless,
our inability of simulating larger quantum systems with classical computers efficiently demands for analytical tools to actually demonstrate strict improvements on complexities of product formulas that apply generally.

\begin{figure}
\includegraphics[angle=90, scale=0.27]{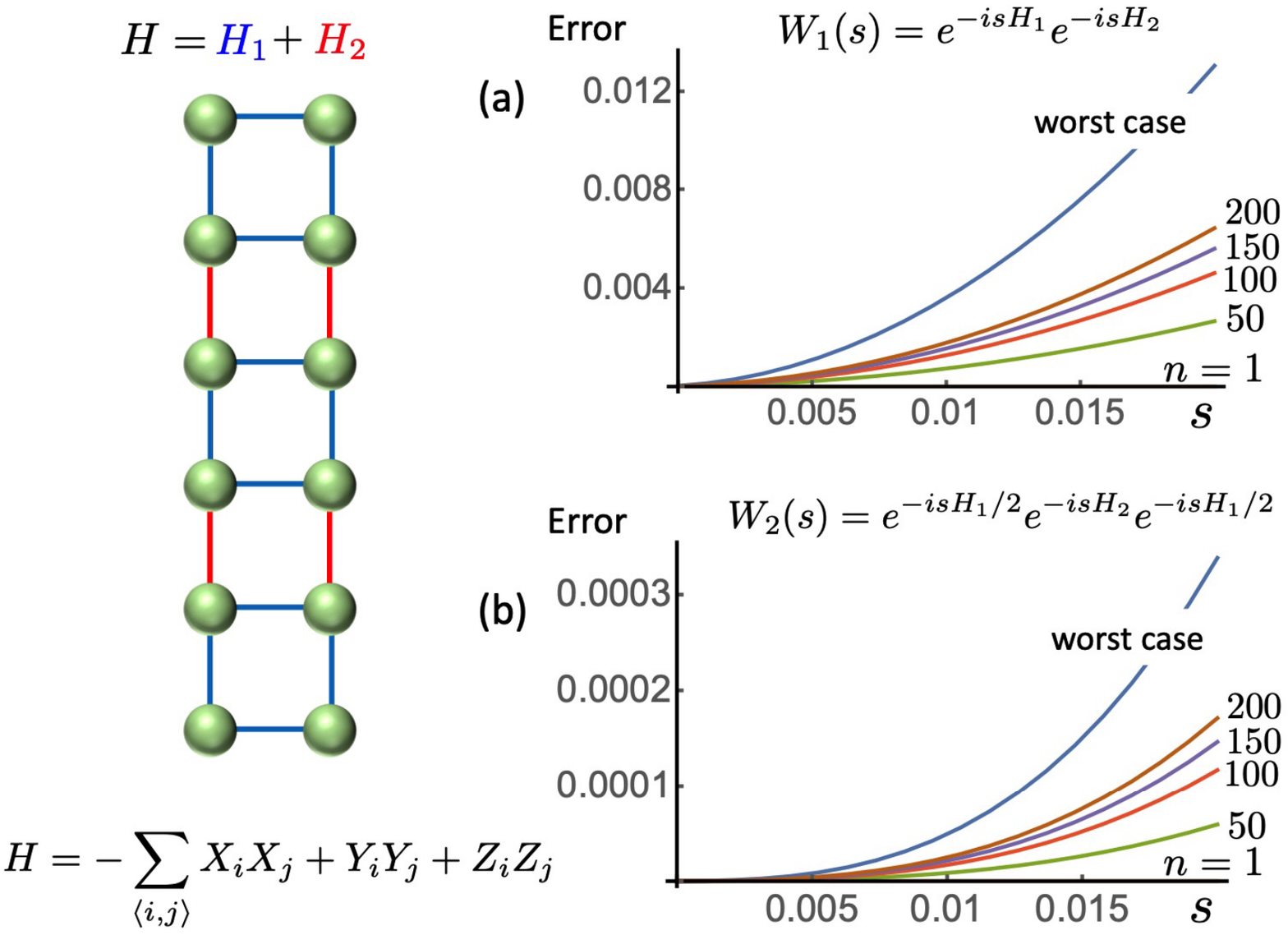}
\caption{Worst case vs. low-energy Trotter errors: Errors induced by product formulas for a 2x6
Heisenberg spin-1/2 model. The Hamiltonian is $H= -\sum_{\langle i, j \rangle} X_i X_j + Y_i Y_j + Z_i Z_j$, where $X_i$, $Y_i$, and $Z_i$ are the Pauli operators for the $i$th spin, and $H=H_1+H_2$, where $H_1$ and $H_2$ are the interaction terms represented by blue and red bonds, respectively.
(a) The evolution operator for time $s$, $U(s)=e^{-isH}$, is approximated by the first order product formula $W_1(s)=e^{-i sH_1}e^{-i sH_2}$. The plot shows the largest approximation errors when acting on various low-energy subspaces associated with increasing energies, labeled by $n=1,50,150,200$, and in the worst case. 
(b) Similar results for when the evolution operator $U(s)=e^{-isH}$ is approximated by the second order product formula $W_2(s)=e^{-i sH_1/2}e^{-i sH_2}e^{-i sH_1/2}$.}
\label{fig:Simulation}
\end{figure}

To this end, we investigate the Hamiltonian simulation problem when the initial state is supported on a {\em low-energy subspace}.
This is a central problem in physics that has vast applications, including the simulation of condensed matter systems for studying quantum phase transitions~\cite{Sac11}, the simulation of quantum field theories~\cite{JLP12}, the simulation of adiabatic quantum state preparation~\cite{FGGS00,BKS10}, and more.
We analyze the complexities of product formulas in this setting and
show significant improvements with respect to the best known complexity bounds that apply to the general case.

\section{Results}

\subsection{Overview}
Our main result is that, for a local Hamiltonian on $N$ spins
$H= \sum_l H_l$  with $H_l \ge 0$,
the error induced by a $p$-th order product formula is $\cO((\Delta' s)^{p+1})$, where $s$ is a (short) time parameter and $\Delta'$ is an effective low-energy norm of $H$.
This norm depends on $\Delta$, which is an energy associated with the initial state, but also depends on $s$ and other parameters that define $H$.
The best known error bounds for product formulas
that apply to the general case depend on the $\|H_l\|$'s~\cite{CST+19}.
(Throughout this paper, $\|.\|$ refers to the spectral norm.) Thus, an improvement in the complexity of product formulas is possible when $\Delta' \ll \max_l \|H_l\|$, which can occur for sufficiently small values of $\Delta$ and $s$. Such values of $s$ appear in low-order product formulas (e.g., first order) or, for larger order, when the overall evolution time $t$ is sufficiently large and/or the desired approximation error $\varepsilon$ is sufficiently small. 
We summarize some of the complexity improvements in
Table~I.

\begin{table}
\begin{tabular}{ |p{0.9cm}||p{2.4cm}|p{3.4cm}| }
\hline
Order & Previous result & Low-energy simulation\\
\hline
$p=1$  & $\cO(\frac{\tau^2 N}{\varepsilon})$    & $ \tilde \cO( \frac{\tau^2}{\varepsilon})+ \cO(\frac{\tau^{4/3} N^{2/3}}{\varepsilon^{1/3}})$ \\
\hline
$p=2$  & $\cO(\frac{\tau^{3/2} N^{1/2}}{\varepsilon^{1/2}} )$ & $ \tilde \cO(\frac{\tau^{3/2}}{\varepsilon^{1/2}}) + \cO( \frac{\tau^{6/5} N^{3/5}}{\varepsilon^{1/5}})$ \\
\hline
$p=3$  & $\cO(\frac{\tau^{4/3} N^{1/3}}{\varepsilon^{1/3}})$& $\tilde \cO(\frac{\tau^{4/3}}{\varepsilon^{1/3}})+ \cO( \frac{\tau^{8/7} N^{4/7}}{\varepsilon^{1/7}})$ \\
\hline
\end{tabular}
\caption{Improvements of low-energy simulation: Comparison between the best-known complexity~\cite{CST+19} and the complexity of low-energy simulation for $p$-th order product formulas. Results show the Trotter number for constant $\Delta$ and local Hamiltonians on $N$ spins with constant degree and strength bounded by $J$, and $\tau=|t|J$. $\varepsilon$ is the approximation error. The $\tilde \cO$ notation hides polylogarithmic factors in $\tau/\varepsilon$.}
\label{fig:Comparison1}
\end{table}

To obtain our results, we introduce the notion of  effective Hamiltonians that are basically the $H_l$'s restricted to act on a low-energy subspace.
The relevant norms of these effective operators is bounded by $\Delta'$.  
One could then proceed to simulate the evolution using a product formula that involves effective Hamiltonians and obtain an error bound that matches ours.
A challenge is that these effective Hamiltonians are generally non-local and difficult to compute. Methods such as the local Schrieffer-Wolff transformation~\cite{glazek1993renormalization, bravyi2011schrieffer} work only at the perturbative regime and numerical renormalization group methods for spin systems~\cite{gu2008tensor, evenbly2009algorithms} have been studied only for a handful of models, while a general analytical treatment does not exist. 
Thus, efficient methods to simulate time evolution of effective Hamiltonians are lacking.
We address this challenge by showing that evolutions under the effective Hamiltonians can be approximated by evolutions under the original $H_l$'s with a suitable choice of $\Delta'$. 
This result is key in our construction and may find applications elsewhere.

Our main contributions are based on a number of technical lemmas and corollaries that are given in the Methods section and proven in detail in Supplementary Information.

\subsection { Product formulas and effective operators} 
For a time-independent Hamiltonian $H=\sum_{l=1}^L H_l$, where each $H_l$ is Hermitian, the evolution operator for time $t$ is $U(t)=e^{-it H}$. 
Product formulas provide a way of approximating $U(t)$ as a product of exponentials, each being a short-time evolution under some $H_l$.
For $p > 0$ integer and $s \in \mathbb R$, a $p$-th order product formula is a unitary
\begin{align}\label{eq:productformula} 
W_p(s)
     =  e^{- i s_q  H_{l_q}}   \cdots   e^{- i s_2  H_{l_2}}  e^{- i s_1  H_{l_1}}\; ,
\end{align}
where each $s_j \in \mathbb R$ is proportional to $s$ and $1 \le l_j \le L$. The number of terms
in the product may depend on $p$ and $L$, and we
assume $L=\cO(1)$, $q=\cO(1)$. (The more general case is analyzed in Supplementary Information.) We define $|{\bf s}|=\sum^{q}_{j=1} |s_j|$ and also assume $|{\bf s}|=\cO(|s|)$. The $p$-th order product formula satisfies $\| U(s) - W_{p}(s) \|=\cO((h|s|)^{p+1})$, where $h=\max_l \|H_l\|$~\cite{BAC07}. 
One way to construct $W_p(s)$ is to apply
a recursion in Refs.~\cite{Suz90,Suz91}.
These are known as Trotter-Suzuki approximations
and satisfy the above assumptions.

By breaking the time interval $t$ into $r$ steps of sufficiently small size $s$,  product formulas
can approximate $U(t)$ as $U(t) \approx (W_{p}(s))^{r}$. We will refer to $r$
as the Trotter number, and this number
will define the complexity of product
formulas that simulate $U(t)$ within given accuracy. 
Note that the total number of terms
in the product formula is actually $qMr=\cO(Mr)$, where $M$ is the number of terms in the product decomposition of each $e^{-is_j H_{l_j}}$.

Known error bounds for product formulas that apply to the general case grow with $h$
and can be large.
However, error bounds for approximating the evolved state $U(t) \ket \psi$ may be better
under the additional assumption that $\ket \psi$ is supported on a low-energy subspace.
We then analyze the case where the initial state satisfies $\Pi_{\le \Delta} \ket \psi = \ket \psi$, where $\Pi_{\le \Lambda}$ is the projector into the subspace spanned by eigenstates of $H$ of energies (eigenvalues) at most $\Lambda \ge 0$. 
We  assume $H_l \ge 0$.
Our results will be especially useful when $\Delta / h$ vanishes asymptotically, and $\Delta$ will specify the low-energy subspace.

The notion of effective operators will be useful in our analysis. Given a Hermitian operator $X$ and $\Delta' \ge \Delta$, the corresponding effective operator is $\bar X = \Pi_{\le \Delta'} X \Pi_{\le \Delta'}$, which is also Hermitian.
We also define the unitaries $\bar U(s) = e^{-is \bar H}$ and $\bar W_p(s)$  by replacing the $H_l$'s by $\bar H_l$'s in $W_p(s)$. 
Note that $\bar h =\max_l \| \bar H_l \| \le \Delta'$ and $U(t) \ket \psi = \bar U(t) \ket \psi$. 
Then, using the known error bound for product formulas, we obtain $\| (U(s) - \bar W_p(s) ) \ket \psi \| = \cO((\Delta' s)^{p+1})$. This error bound is a significant improvement over the general case if $\Delta' \ll h$, which may occur when $\Delta \ll h$. 
However, product approximations of $U(t)$ require that each term is an exponential of some $H_l$, which is not the case in $\bar W_p(s)$. 
We will address this issue and show that the improved error bound is indeed attained by $W_p(s)$ for  a suitable $\Delta'$.

\subsection {Local Hamiltonians}
We are interested in simulating the time evolution of a local $N$-spin system on a lattice. 
Each local interaction term in $H$ is of  strength bounded by $J$ and involves, at most, $k$ spins. 
We do not assume that these interactions are only within neighboring spins but define the degree $d$ as the maximum number of local interaction terms that involve any spin.
Next, we write $H=\sum_{l=1}^L H_l$, where each $H_l$ is a sum of $M$ local, commuting  terms~\cite{Bro41} and $LM \le dN$.
Each $e^{-i sH_l}$ in a product formula can be decomposed as products of $M$ evolutions under the local, commuting  terms with no error.

These local Hamiltonians appear as important condensed matter systems, including gapped and critical spin chains, topologically ordered systems, and models with long-range interactions~\cite{LSM61,AKLT04,Kit03,LMG65}.
For example, for a spin chain with nearest neighbour interactions, $L=2$ and each $H_l$ may refer to interaction terms associated with even and odd bonds, respectively. In this case, $h=\cO(N)$.
We will present our results for the case $k=\cO(1)$ and $d=\cO(1)$ in the main text, which further imply $L=\cO(1)$~\cite{Bro41}. Nevertheless, explicit dependencies of our results in $k$, $d$, $L$, and other parameters that specify $H$ can be found in Supplementary Note 4.

Table~II summarizes the relevant parameters for the simulation of local Hamiltonians with product formulas. 

\begin{table}
\begin{tabular}{ |p{1.43cm}||p{4.5cm}| }
\hline
Symbol & Meaning \\
\hline
$J$  & Hamiltonian term strength \\
\hline
$\Delta$  & Low-energy parameter\\
\hline
$\Delta'$ ($\geq \Delta$)  & Effective low-energy norm\\
\hline
$N$  & Number of spins\\
\hline
$t$  & Total evolution time\\
\hline
$r$  & Number of Trotter steps\\
\hline
$s$  & $=t/r$, unit Trotter time\\
\hline
$\varepsilon$  & Total simulation error\\
\hline
\end{tabular}
\caption{Notation: The parameters of the Hamiltonian (local and constant-degree) and product formula simulation.}
\label{fig:Notation}
\end{table}

\subsection{Main result}

\begin{theorem}
\label{thm:mainthm}
Let $H=\sum_{l=1}^L H_l$ be a $k$-local Hamiltonian as above, $H_l \ge 0$, $\Delta \ge 0$, $0 \le J |s| \le 1$, and $W_p(s)$ a $p$-th order product formula as in Eq.~\eqref{eq:productformula}. Then,
\begin{align}
\label{eq:maineq}
    \| (U(s) - W_p(s)) \Pi_{\le \Delta} \| = \cO((\Delta' s)^{p+1}) \;,
\end{align}
where $\Delta' =\Delta + \beta_1 J \log( \beta_2/(J|s|)) + \beta_3 J^2 N |s|$ and the $\beta_i$'s are positive constants,
$\beta_2 \ge 1$.
\end{theorem}

The proof of Thm.~\ref{thm:mainthm}
is in Supplementary Note 3 and we provide more details about it in the next section, but the basic idea is as follows. 
There are two contributions to Eq.~\eqref{eq:maineq} in our analysis. One comes from approximating the evolution operator with a product formula that involves the effective Hamiltonians and, as long as $\Delta' \ge \Delta$, this error is $\cO((\Delta' |s|)^{p+1})$, as explained.
The other comes from replacing such a product formula by the one with the actual Hamiltonians $H_l$, i.e., $W_p(s)$. 
However, unlike $\bar H_l$, the evolution under each $H_l$ allows for leakage or transitions from the low-energy subspace to the subspace of energies higher than $\Delta'$. 
In Supplementary Information, we use a result on energy distributions in Ref.~\cite{AKL16} to show that this leakage can be bounded and decays exponentially with $\Delta'$. 
Thus, this effective norm depends on $\Delta$ and must also depend on $s$, as the support on high-energy states can increase 
as $s$ increases, resulting in the linear contribution to $\Delta'$ in Thm.~\ref{thm:mainthm}. 

The $\log (\beta_2/(J|s|))$ factor in $\Delta'$
only becomes relevant when $|s| \ll 1$. 
This term appears in our analysis due to the requirement that both contributions to Eq.~\eqref{eq:maineq} discussed above are of the same order.
Thus, as $s \rightarrow 0$, we require $\Delta ' \rightarrow \infty$ to make the error due to leakage zero, which is unnecessary and unrealistic.
This term plays a mild role when determining the final complexity of a product formula, as the goal will be to make $s$ as large as possible for a target approximation error.
It may be possible that this term disappears in a more refined analysis.

Let $r=t/s$ be the Trotter number, i.e., the number
of steps to approximate $U(t)$ as $(W_p(s))^r$.
Since $U(s)\Pi_{\le \Delta}=\Pi_{\le \Delta} U(s)\Pi_{\le \Delta}$ and if $\|(U(s) - W_p(s)) \Pi_{\le \Delta}\| \le \epsilon$, the triangle inequality implies $\|(U(t) - (W_p(s))^r) \Pi_{\le \Delta}\| \le 2 r \epsilon$.
Thus, for overall target error $\varepsilon >0$, it will suffice to satisfy
$\|(U(s) - W_p(s)) \Pi_{\le \Delta}\| = \cO(\varepsilon s/t)$.
This condition and Thm.~\ref{thm:mainthm} can be used to determine $r$ as follows.

Each term of $\Delta'$ in Thm.~\ref{thm:mainthm} can be dominant depending on $s$ and $\Delta$. First, we consider the first two terms, and determine a condition in $s$ to satisfy $((\Delta +J) |s|)^{p+1} = \cO(\varepsilon s/t)$, by omitting the $\log$ factor. 
Then, we consider another term and determine a condition in $s$ to satisfy $(J^2 N |s|^2)^{p+1} = \cO(\varepsilon s/t)$.
These two conditions alone can be satisfied
with a Trotter number
\begin{align}
\label{eq:rbound0}
    r'= \cO \left( \frac{(t (\Delta + J)  )^{1+\frac{1}{p}}} {\varepsilon^{\frac{1}{p}}}+ \frac{(tJ \sqrt{ N })^{1+\frac{1}{2p+1}}} {\varepsilon^{ \frac{1}{2p+1}}}\right) \;.
\end{align}
Last, we reconsider the second term with $\log$, and we require  $(J \log(1/(J|s|)) |s|)^{p+1} = \cO(\varepsilon s/t)$. 
As the first two conditions are satisfied with a value for $s$ that is polynomial in $N$ and $ tJ/\varepsilon$, this last condition only sets a correction to the first term in $r'$ in Eq.~\eqref{eq:rbound0} that is polylogarithmic in $ |t|J/\varepsilon$. 
Thus, the overall complexity of the product formula for local Hamiltonians is given by Eq.~\eqref{eq:rbound0}, where we need to replace $\cO$ by $\tilde \cO$ to account for the last correction. 
Note that the number of terms in each $W_p(s)$ is constant under the assumptions and $r$ is proportional to the total number of exponentials in $(W_p(s))^r$.

We give a general result on the complexity
of product formulas that provides $r$ as a function of all parameters that specify $H$ in
Thm.~2 of Supplementary Note 4.

\subsection{The condition $H_l \ge 0$}

The error bounds for product formulas used in Thm.~\ref{thm:mainthm} depend
on the norm of the effective Hamiltonians $\bar H_l$.
The assumption $H_l \ge 0$ will then assure
that $\|\bar H_l \| \le \Delta'$, which is sufficient 
to demonstrate the complexity improvements in Eq.~\eqref{eq:rbound0}.

In general, $H_l \ge 0$ can be met after a simple shifting $H_l \rightarrow H_l+a_l$, and the assumption seems irrelevant. However, this shifting could result in a value of $\Delta$ (or $\Delta'$) that scales with some parameters such as the system size $N$. In this case, the error bound in Thm.~\ref{thm:mainthm} would be comparable to that of the worst case (without the low-energy assumption) and would not provide
an advantage.

Nevertheless, for many important spin Hamiltonians, the assumption $H_l \ge 0$ is readily satisfied.
The Heisenberg model of Fig.~1 is an example, where $H_l$ is a sum of terms like $\one - X_iX_j -Y_iY_j- Z_iZ_j \ge 0$. More general (anisotropic) Heisenberg models
as well as the so-called frustration-free Hamiltonians that are ubiquitous in many-body physics also satisfy the assumption~\cite{BT09,BOO10}, where our results directly apply. For this class of models,
the ground state energy is zero. This class contains interesting low-lying states in the subspace where, e.g., $\Delta=\cO(1)$.

We provide more details on potential complexity improvements for the general case ($H_l \ngeq 0$) at the end of Supplementary Note 3.

\vspace{0.2cm}


\section{Discussions}


The best previous result for the complexity of product
formulas (Trotter number) for local Hamiltonians of constant degree is $\cO(\tau^{1+1/p} N^{1/p}/\varepsilon^{1/p})$, with $\tau = |t|J$~\cite{CST+19}.
Our result gives an improvement over this in various regimes. 
Note that, a general characteristic of our results is that they depend on $\Delta$, which is specified by the initial state.
Here we assume that $\Delta$ is a constant independent of other parameters that specify $H$. The comparison for this case is in Table~I.
For $p=1$, we obtain a strict improvement of order $N^{1/3}$ over the best-known result. For higher values of $p$, the improvement appears for larger values of $\tau /\varepsilon$ that may scale with $N$, e.g., $\tau /\varepsilon=\tilde \Omega( N^{p - 2 +1/(p+1)})$.
In Supplementary Note 5, we provide a more detailed comparison between our results and the best previous results for product formulas as a function of $\Delta$ and other parameters that specify $H$.

A more recent method for Hamiltonian simulation uses a truncated Taylor series expansion of $e^{-iHt/r} \approx U_r=\sum^{K}_{k=0} (-iHt/r)^k/k!$~\cite{BCC+15}.
Here, $r$ is the number of ``segments'', and $U(t)$ is approximated as $(U_r)^r$.
A main advantage of this method is that, unlike product formulas, its complexity in terms of $\varepsilon$ is logarithmic, a major advantage if precise computations are needed. 
The complexity of this method for the low-energy subspace of $H$ can only be mildly improved. A small $\Delta$ allows for a truncation value $K$ that is smaller than that for the general case~\cite{BCC+15}. 
Nevertheless, the complexity of this method is dominated by $r$, which depends on a certain 1-norm $\|H\|_1$ of $H$ that is independent of $\Delta$. 
Furthermore, quantum signal processing,
an approach for Hamiltonian simulation also based on
certain polynomial approximations of $U(t)$,
was recently considered for simulation in the low-energy subspace~\cite{low2017hamiltonian}. 
While the low-energy constraint may also result
in some mild (constant) improvement, the overall complexity of quantum signal processing also depends on $\|H\|_1$. For local Hamiltonians where $k,d=\cO(1)$, and for constant $\Delta$, the overall complexity of these methods is $\tilde \cO(\tau N^2)$, where we disregarded logarithmic factors in $\tau$, $N$, and $1/\varepsilon$.
Our results on product formulas provide an improvement over these methods in various regimes, e.g., when $\varepsilon$ is constant.

\vspace{0.2cm}

The obtained complexities are an improvement as long as the energy $\Delta$ of the initial state is sufficiently small.
As we discussed, the assumption $H_l \geq 0$ was used and, while our results readily apply to a large class of spin models, it may be in conflict with ensuring small values of $\Delta$ in some cases.
It will be important to understand this in more detail (see Supplementary Note 3),
which may be related to the fact that, for general Hamiltonians ($H_l \ngeq 0$), an improvement in the low-energy simulation could imply an improvement in the high-energy simulation by considering $-H$ instead. Indeed, certain spin models possess a symmetry that connects the high-energy and low-energy subspaces via a simple transformation. 
Whether such ``high-energy'' simulation improvement is possible or not remains open.
Additionally, known complexities of product formulas are polynomial in $1/\varepsilon$. This is an issue if precise computations are required as in the case of quantum field theories or QED.
Whether this complexity can be improved in terms of precision as in Refs.~\cite{BCC+14,BCC+15,LC17, low2019well} is also open.

Our work is an initial attempt to this problem. We expect to motivate further studies on improved Hamiltonian simulation methods in this setting by refining our analyses, assuming other structures such as interactions that are geometrically local, or improving other simulation approaches.


\section{Methods} 

\subsection{Leakage to high-energy states}
A key ingredient for Thm.~\ref{thm:mainthm} is a property
of local spin systems, where the  leakage
to high-energy states due to the evolution under any $H_l$ can be bounded. Let $\Pi_{>\Lambda'}$ be the projector into the subspace spanned by eigenstates of energies greater than $\Lambda'$. Then, for a state $\ket \phi$ that
satisfies $\Pi_{\le \Lambda} \ket \phi = \ket \phi$,
we consider a question on the support of $e^{-isH_l} \ket \phi$ on states with energies greater than $\Lambda'$.
This question arises naturally in Hamiltonian complexity and beyond, and Lemma~\ref{lem:TimeLeak} below may be of independent interest. 
A generalization of this lemma will allow one to address the Hamiltonian simulation problem in the low-energy subspace
beyond spin systems.

\begin{lemma}[Leakage to high energies]
\label{lem:TimeLeak}
Let $H =\sum_{l=1}^L H_l$ be a $k$-local Hamiltonian of constant degree as above, $H_l \ge 0$, and $\Lambda' \ge \Lambda \ge 0$. Then, 
$\forall \; s \in \mathbb R$ and $\forall \; l$,
\begin{align}
\label{eq:TimeLeak}
     \| \Pi_{ > \Lambda'} e^{-i s H_{l}} \Pi_{\le \Lambda} \|  \le e^{- \alpha_1 (\Lambda' - \Lambda)/J} \left( e^{\alpha_2 J |s| N} - 1 \right) \;,
\end{align}
where $\alpha_1$ and $\alpha_2$ are positive constants. 
\end{lemma}
The proof is in Supplementary Note 1. 
It follows from a result in Ref.~\cite{AKL16} on the action of a local interaction term on a quantum state of low-energy, in combination with a series expansion of $e^{-is H_l}$.
While the local interaction term could generate support on arbitrarily high-energy states, that support is suppressed by a factor that decays exponentially in $\Lambda'-\Lambda$. 

Another key ingredient for proving Thm.~\ref{thm:mainthm} is the ability to replace evolutions under the $H_l$'s in a product formula by those under their effective low-energy versions (and vice versa) with bounded error. 
This is addressed by Lemma~\ref{lem:EffTimeLeak} below, which is a consequence of Lemma~\ref{lem:TimeLeak}. 
The proof is in  Supplementary Note 2, where we also provide tighter bounds that depend on $\Delta'$.

\begin{lemma}
\label{lem:EffTimeLeak}
Let $H=\sum_{l=1}^L H_l$ be a $k$-local Hamiltonian of constant degree as above, $H_l \ge 0$, and $\Delta' \ge \Lambda' \ge \Lambda \ge 0$. Then, $\forall \; s \in \mathbb R$ and $\forall \; l$,
\begin{align}
\nonumber
     \| \Pi_{\le \Lambda'} (e^{-i s \bar H_{l}} &-  e^{-i s H_{l}})  \Pi_{\le \Lambda} \| \\
    & \le e^{-\alpha_1 (\Lambda' - \Lambda)/J} (e^{\alpha_2 J |s| N} - 1)
\end{align}
and
\begin{align}
   \| \Pi_{>\Lambda'} e^{-is \bar H_{l}} \Pi_{\le \Lambda } \| \le 3 e^{-\alpha_1 (\Lambda' - \Lambda)/J} (e^{\alpha_2 J |s| N} - 1)  \;,
\end{align}
where $\alpha_1$ and $\alpha_2$ are positive constants.
\end{lemma}

\subsection{Relevance to the main result}

The consequences of these lemmas for Hamiltonian simulation
are many-fold and we only sketch
those that are relevant for
Thm.~\ref{thm:mainthm}.
Consider any product formula of the form
$W=\prod_{j=1}^q e^{-i s_j H_{l_j}}$.
Then, there exists a sequence of energies $\Lambda_q \ge \ldots \ge \Lambda_0=\Delta$ such that the action of $W$ on the initial low-energy state $\ket \psi$ can be well approximated by that of $ W^\bl=\prod_{j=1}^q \Pi_{\le \Lambda_j}e^{-i s_j H_{l_j}}$ on the same state.  
Furthermore, each $\Pi_{\le \Lambda_j}e^{-i s_j H_{l_j}}$ in $W^\bl$ can be replaced by $\Pi_{\le \Lambda_j}e^{-i s_j \bar H_{l_j}}$ and later by $e^{-i s_j \bar H_{l_j}}$ within the same error order, as long as $\Lambda_q \le \Delta'$.

In particular, for sufficiently small evolution times $s_j$ and $\Delta \ll h$, the resulting effective norm
satisfies $\Delta' \ll h$ for local Hamiltonians.
This is formalized by several corollaries in Supplementary Note 3.
Starting from $W$, we can construct
the product formula $\bar W = \prod_{j=1}^q e^{-i s_j \bar H_{l_j}}$. Lemmas~\ref{lem:TimeLeak} and~\ref{lem:EffTimeLeak} imply that both product formulas
produce approximately the same state when acting on $\ket \psi$, for a suitable choice of $\Delta'$
as in Thm.~\ref{thm:mainthm}. If $\bar W$ is a product
formula approximation of $\bar U(s)=e^{-is \bar H}$, it follows that $U(s) \ket \psi = \bar{U}(s) \ket \psi \approx \bar W \ket \psi \approx W \ket{\psi}$. 

\vspace{0.2 cm}


\section{Data Availability statement}
All relevant data used for Fig.~1 are available from the authors.

\section{Code Availability statement}
The code for the simulation results in Fig.~1 is available from the authors.
\section{Acknowledgements}
We acknowledge support from the LDRD program at LANL, the U.S. Department of Energy, Office of Science, High Energy Physics and Office of Advanced Scientific Computing Research, under the QuantISED grant KA2401032 and Accelerated Research in Quantum Computing (ARQC). Los Alamos National Laboratory is managed by Triad National Security, LLC, for the National Nuclear Security Administration of the U.S. Department of Energy under Contract No. 89233218CNA000001. 
This work is also supported by the Quantum Science Center (QSC), a National Quantum Information Science Research Center of the U.S. Department of Energy (DOE).\\

\section{Author Contributions}

All authors contributed equally to this work.

\section{Competing Interests}
The authors declare no competing interests.




\section{Figure legends}

Figure 1: Worst case vs. low-energy Trotter errors: Errors induced by product formulas for a 2x6
Heisenberg spin-1/2 model. The Hamiltonian is $H= -\sum_{\langle i, j \rangle} X_i X_j + Y_i Y_j + Z_i Z_j$, where $X_i$, $Y_i$, and $Z_i$ are the Pauli operators for the $i$th spin, and $H=H_1+H_2$, where $H_1$ and $H_2$ are the interaction terms represented by blue and red bonds, respectively.
(a) The evolution operator for time $s$, $U(s)=e^{-isH}$, is approximated by the first order product formula $W_1(s)=e^{-i sH_1}e^{-i sH_2}$. The plot shows the largest approximation errors when acting on various low-energy subspaces associated with increasing energies, labeled by $n=1,50,150,200$, and in the worst case. (b) Similar results for when the evolution operator $U(s)=e^{-isH}$ is approximated by the second order product formula $W_2(s)=e^{-i sH_1/2}e^{-i sH_2}e^{-i sH_1/2}$.\\

Table I: Improvements of low-energy simulation: Comparison between the best-known complexity~\cite{CST+19} and the complexity of low-energy simulation for $p$-th order product formulas. Results show the Trotter number for constant $\Delta$ and local Hamiltonians on $N$ spins with constant degree and strength bounded by $J$, and $\tau=|t|J$. $\varepsilon$ is the approximation error. The $\tilde \cO$ notation hides polylogarithmic factors in $\tau/\varepsilon$.\\

Table II: Notation: The parameters of the Hamiltonian (local and constant-degree) and product formula simulation.

\bibliographystyle{unsrt}
\bibliography{Subspace-Sim}

\begin{thebibliography}{10}

\bibitem{Fey82}
Richard~P. Feynman.
\newblock Simulating physics with computers.
\newblock {\em International Journal of Theoretical Physics},
  21(6--7):467--488, 1982.

\bibitem{Llo96}
Seth Lloyd.
\newblock Universal quantum simulators.
\newblock {\em Science}, 273(5278):1073--1078, (1996).

\bibitem{AT03}
Dorit Aharonov and Amnon Ta-Shma.
\newblock Adiabatic quantum state generation and statistical zero knowledge.
\newblock In {\em Proceedings of the thirty-fifth annual ACM symposium on
  Theory of computing}, pages 20--29, (2003).

\bibitem{BAC07}
Dominic~W. Berry, Graeme Ahokas, Richard Cleve, and Barry~C. Sanders.
\newblock Efficient quantum algorithms for simulating sparse hamiltonians.
\newblock {\em Communications in Mathematical Physics}, 270(2):359--371,
  (2007).

\bibitem{WBH+10}
Nathan Wiebe, Dominic Berry, Peter H{\o}yer, and Barry~C. Sanders.
\newblock Higher order decompositions of ordered operator exponentials.
\newblock {\em Journal of Physics A: Mathematical and Theoretical},
  43(6):065203, (2010).

\bibitem{BCC+14}
Dominic~W. Berry, Andrew~M. Childs, Richard Cleve, Robin Kothari, and
  Rolando~D. Somma.
\newblock Exponential improvement in precision for simulating sparse
  hamiltonians.
\newblock In {\em Proceedings of the forty-sixth annual ACM symposium on Theory
  of computing}, pages 283--292, 2014.

\bibitem{BCC+15}
Dominic~W. Berry, Andrew~M. Childs, Richard Cleve, Robin Kothari, and
  Rolando~D. Somma.
\newblock Simulating hamiltonian dynamics with a truncated taylor series.
\newblock {\em Physical review letters}, 114(9):090502, (2015).

\bibitem{LC17}
Guang~Hao Low and Isaac~L. Chuang.
\newblock Optimal hamiltonian simulation by quantum signal processing.
\newblock {\em Physical review letters}, 118(1):010501, (2017).

\bibitem{LC19}
Guang~Hao Low and Isaac~L. Chuang.
\newblock Hamiltonian simulation by qubitization.
\newblock {\em Quantum}, 3:163, (2019).

\bibitem{Cam19}
Earl Campbell.
\newblock A random compiler for fast hamiltonian simulation.
\newblock {\em Physical review letters}, 123:070503, (2019).

\bibitem{haah2018quantum}
Jeongwan Haah, Matthew Hastings, Robin Kothari, and Guang~Hao Low.
\newblock Quantum algorithm for simulating real time evolution of lattice
  hamiltonians.
\newblock In {\em 2018 IEEE 59th Annual Symposium on Foundations of Computer
  Science (FOCS)}, pages 350--360. IEEE, (2018).

\bibitem{SOGKL02}
Rolando Somma, Gerardo Ortiz, James~E. Gubernatis, Emanuel Knill, and Raymond
  Laflamme.
\newblock Simulating physical phenomena by quantum networks.
\newblock {\em Physical Review A}, 65(4):042323, (2002).

\bibitem{JLP12}
Stephen~P. Jordan, Keith~S.M. Lee, and John Preskill.
\newblock Quantum algorithms for quantum field theories.
\newblock {\em Science}, 336:1130, (2012).

\bibitem{WBCH14}
Dave Wecker, Bela Bauer, Bryan~K. Clark, Matthew~B. Hastings, and Matthias
  Troyer.
\newblock Gate-count estimates for performing quantum chemistry on small
  quantum computers.
\newblock {\em Physical Review A}, 90(2):022305, (2014).

\bibitem{BBK+15}
Ryan Babbush et~al.
\newblock Exponentially more precise quantum simulation of fermions i: Quantum
  chemistry in second quantization.
\newblock {\em New Journal of Physics}, 18(3):033032, (2016).

\bibitem{BBK++15}
Ryan Babbush et~al.
\newblock Exponentially more precise quantum simulation of fermions in the
  configuration interaction representation.
\newblock {\em Quantum Science and Technology}, 3(1):015006, (2017).

\bibitem{CKS17}
Andrew~M. Childs, Robin Kothari, and Rolando~D. Somma.
\newblock Quantum algorithm for systems of linear equations with exponentially
  improved dependence on precision.
\newblock {\em SIAM Journal on Computing}, 46(6):1920--1950, (2017).

\bibitem{Suz90}
Masuo Suzuki.
\newblock Fractal decomposition of exponential operators with applications to
  many-body theories and monte carlo simulations.
\newblock {\em Physics Letters A}, 146(6):319--323, (1990).

\bibitem{Suz91}
Masuo Suzuki.
\newblock General theory of fractal path integrals with applications to
  many-body theories and statistical physics.
\newblock {\em Journal of Mathematical Physics}, 32(2):400--407, (1991).

\bibitem{NB98}
Mark~E.J. Newman and Gerard~T. Barkema.
\newblock {\em Monte Carlo Methods in Statistical Physics}.
\newblock Oxford University Press, (1998).

\bibitem{Som16}
Rolando~D. Somma.
\newblock A trotter-suzuki approximation for lie groups with applications to
  hamiltonian simulation.
\newblock {\em Journal of Mathematical Physics}, 57(6):062202, (2016).

\bibitem{CS19}
Andrew~M. Childs and Yuan Su.
\newblock Nearly optimal lattice simulation by product formulas.
\newblock {\em Physical review letters}, 123:050503, (2019).

\bibitem{CST+19}
Andrew~M. Childs, Yuan Su, Minh~C. Tran, Nathan Wiebe, and Shuchen Zhu.
\newblock Theory of trotter error with commutator scaling.
\newblock {\em Physical Review X}, 11(1):011020, (2021).

\bibitem{CBC20}
Laura Clinton, Johannes Bausch, and Toby Cubitt.
\newblock Hamiltonian simulation algorithms for near-term quantum hardware.
\newblock {\em Preprint at
  \href{https://arxiv.org/abs/2003.06886}{https://arxiv.org/abs/2003.06886}},
  (2020).

\bibitem{Sac11}
Subir Sachdev.
\newblock {\em Quantum Phase Transitions, 2nd Edition}.
\newblock Cambridge University Press, Cambridge, (2011).

\bibitem{FGGS00}
Edward Farhi, Jeffrey Goldstone, Sam Gutmann, and Michael Sipser.
\newblock Quantum computation by adiabatic evolution.
\newblock {\em Preprint at
  \href{https://arxiv.org/abs/quant-ph/0001106}{https://arxiv.org/abs/quant-ph/0001106}},
  2000.

\bibitem{BKS10}
Sergio Boixo, Emanuel Knill, and Rolando~D. Somma.
\newblock Fast quantum algorithms for traversing paths of eigenstates.
\newblock {\em Preprint at
  \href{https://arxiv.org/abs/1005.3034}{https://arxiv.org/abs/1005.3034}},
  2010.

\bibitem{glazek1993renormalization}
Stanis{\l}aw~D. G{\l}azek and Kenneth~G. Wilson.
\newblock Renormalization of hamiltonians.
\newblock {\em Physical Review D}, 48(12):5863, (1993).

\bibitem{bravyi2011schrieffer}
Sergey Bravyi, David~P DiVincenzo, and Daniel Loss.
\newblock Schrieffer--wolff transformation for quantum many-body systems.
\newblock {\em Annals of physics}, 326(10):2793--2826, (2011).

\bibitem{gu2008tensor}
Zheng-Cheng Gu, Michael Levin, and Xiao-Gang Wen.
\newblock Tensor-entanglement renormalization group approach as a unified
  method for symmetry breaking and topological phase transitions.
\newblock {\em Physical Review B}, 78(20):205116, (2008).

\bibitem{evenbly2009algorithms}
Glen Evenbly and Guifr{\'e} Vidal.
\newblock Algorithms for entanglement renormalization.
\newblock {\em Physical Review B}, 79(14):144108, (2009).

\bibitem{Bro41}
Rowland~Leonard Brooks.
\newblock On colouring the nodes of a network.
\newblock In {\em Mathematical Proceedings of the Cambridge Philosophical
  Society}, volume~37, pages 194--197. Cambridge University Press, (1941).

\bibitem{LSM61}
Elliott Lieb, Theodore Schultz, and Daniel Mattis.
\newblock Two soluble models of an antiferromagnetic chain.
\newblock {\em Annals of Physics}, 16(3):407--466, (1961).

\bibitem{AKLT04}
Ian Affleck, Tom Kennedy, Elliott~H. Lieb, and Hal Tasaki.
\newblock Rigorous results on valence-bond ground states in antiferromagnets.
\newblock In {\em Condensed Matter Physics and Exactly Soluble Models}, pages
  249--252. Springer, (2004).

\bibitem{Kit03}
A.~Yu Kitaev.
\newblock Fault-tolerant quantum computation by anyons.
\newblock {\em Annals of Physics}, 303(1):2--30, (2003).

\bibitem{LMG65}
H.~J. Lipkin, N.~Meshkov, and A.J. Glick.
\newblock Validity of many-body approximation methods for a solvable model:(i).
  exact solutions and perturbation theory.
\newblock {\em Nuclear Physics}, 62(2):188--198, (1965).

\bibitem{AKL16}
Itai Arad, Tomotaka Kuwahara, and Zeph Landau.
\newblock Connecting global and local energy distributions in quantum spin
  models on a lattice.
\newblock {\em Journal of Statistical Mechanics: Theory and Experiment},
  2016(3):033301, (2016).

\bibitem{BT09}
Sergey Bravyi and Barbara Terhal.
\newblock Complexity of stoquastic frustration-free hamiltonians.
\newblock {\em Siam journal on computing}, 39(4):1462--1485, (2010).

\bibitem{BOO10}
Niel de~Beaudrap, Matthias Ohliger, Tobias~J. Osborne, and Jens Eisert.
\newblock Solving frustration-free spin systems.
\newblock {\em Physical review letters}, 105:060504, (2010).

\bibitem{low2017hamiltonian}
Guang~Hao Low and Isaac~L. Chuang.
\newblock Hamiltonian simulation by uniform spectral amplification.
\newblock {\em Preprint at
  \href{https://arxiv.org/abs/1707.05391}{https://arxiv.org/abs/1707.05391}},
  (2017).

\bibitem{low2019well}
Guang~Hao Low, Vadym Kliuchnikov, and Nathan Wiebe.
\newblock Well-conditioned multiproduct hamiltonian simulation.
\newblock {\em Preprint at
  \href{https://arxiv.org/abs/1907.11679}{https://arxiv.org/abs/1907.11679}},
  (2019).

\end{thebibliography}


\onecolumngrid

\newpage



\begin{center}
	\textbf{\large Supplementary Information}
\end{center}


In the following, we let $H$
be a $k$-local Hamiltonian of $N$ spins, where each interaction term involves at most $k>0$ spins. 
We will write $H=\sum_{l=1}^L H_l$, where each $H_l$ is a sum of at most $M$ $k$-local and commuting interaction terms.
The $H_l$'s may be obtained 
via graph coloring~\cite{Bro41}, where a graph can be constructed with vertices that are labeled according to the subset of spins in each interaction term and the edges connect vertices associated with the same spins, but more efficient constructions may be possible. Indeed, in many interesting examples such as spins on the square lattice, $H$ is already given in the desired form.
The degree of $H$, i.e., the highest number of interaction terms that act non-trivially on any spin, is $d>0$. 
The strength of each local interaction term is bounded by $J>0$, hence $\|H_l\| \le JM$ and $\|H\| \le JML$. 
Throughout this paper, $\|.\|$ refers to the spectral norm (largest eigenvalue for positive semidefinite and Hermitian operators).
The total number of local terms in $H$ is then upper bounded by $ML$ and $dN$. We will assume that $H$ contains exactly $ML$ terms and thus $N \le ML \le dN$ with no loss of generality (e.g., we can add or subtract trivial terms to $H_l$ and each spin appears, at least, in one term). Furthermore, following the coloring procedure described above, we may assume $L \le kd +1$~\cite{Bro41}.

For $\Lambda' \ge \Lambda \ge 0$, the operators $\Pi_{\le \Lambda}$ and $\Pi_{>\Lambda'}$
are the projectors into the subspaces spanned by the eigenstates of $H$ with energies (eigenvalues) at most $\Lambda$ and larger than $\Lambda'$, respectively.
For a given $\Delta '\ge 0$,  the $\Delta'$-effective (or simply effective) Hamiltonians are then defined as $\bar H = \Pi_{\le \Delta '} H \Pi_{\le \Delta'}$,  $\bar H_{l} = \Pi_{\le \Delta '} H_{l} \Pi_{\le \Delta'}$, and $\bar H=\sum_{l=1}^L \bar H_{l}$.
We assume $H_{l} \ge 0$, and thus $\bar H_l \ge 0$, $\| \bar H_{l} \| \le \| \bar H \| = \Delta'$. 

\vspace{0.1cm}

\section*{Supplementary Notes}

\subsection*{Supplementary Note 1: Proof of Lemma~1}
We employ Theorem 2.1 in Ref.~\cite{AKL16}
that, for an operator $A$, states
\begin{align}
\|\Pi_{>\Lambda'} A  \Pi_{\leq \Lambda}\| \leq \|A\| \cdot e^{-\lambda (\Lambda' - \Lambda - 2R)} \;.
\end{align}
Here, $\lambda= 1/(2gk)$, where $g$ is an upper bound on the sum of the strengths of the interactions associated with any spin. 
In our case, we take $g=dJ$ and $\lambda=1/(2Jdk)$.
If $E_A$ is the subset of local interaction terms in $H$ that do not commute with $A$, $R$ is the sum of the strengths of the terms in $E_A$. 
For any $H_l$, we note that $(H_l)^n$ is a sum of, at most, $M^n$ terms, each of strength bounded by $J^n$ and containing, at most, $kn$ spins. 
For each such term, $R \le Jdkn$, and we obtain
\begin{align}
\| \Pi_{ > \Lambda'} (H_l)^n \Pi_{\le \Lambda} \| & \le (M J)^n  e^{-\frac{1}{2Jdk} (\Lambda' - \Lambda -2 Jdkn)} \\
\label{eq:HLenergyaction}
\; & = (e M J)^n  e^{-\frac{1}{2Jdk} (\Lambda' - \Lambda)} \;.
\end{align}

We now consider the Taylor series expansion of the exponential,
\begin{align}
e^{-i s H_l} = \sum_{n=0}^\infty \frac {(-is H_l)^n}{n!} \;.
\end{align}
The triangle inequality and Eq.~\eqref{eq:HLenergyaction} imply
\begin{align}
\label{eq:lemma3Eq1}
\| \Pi_{ > \Lambda'} e^{-i s H_l} \Pi_{\le \Lambda} \|  & \le \sum_{n=1}^\infty \frac {|s|^n \| \Pi_{ > \Lambda'} (H_l)^n \Pi_{\le \Lambda} \|} {n!} \\
& \le e^{-\frac{1}{2Jdk}(\Lambda' - \Lambda)} \sum_{n=1}^\infty \frac {(|s| e M J)^n }{n!}\\
& = e^{-\frac{1}{2Jdk} (\Lambda' - \Lambda)} \left( e^{|s| eM J } - 1 \right)
\\
\label{eq:lemma3Eq2}
& = e^{-\lambda (\Lambda' - \Lambda)} \left( e^{\alpha |s| M} - 1 \right) \;,
\end{align}
where $\alpha = e J$.

In particular, when $k$ and $d$ are constants,
we have $\alpha M \le e J d N$, and Eq.~\eqref{eq:lemma3Eq2} implies
\begin{align}
    \| \Pi_{ > \Lambda'} e^{-i s H_l} \Pi_{\le \Lambda} \|  & \le
    e^{-\alpha_1 (\Lambda' - \Lambda)/J} \left( e^{\alpha_2 J |s| N} - 1 \right) \;,
\end{align}
where $\alpha_1=1/(2dk)$ and $\alpha_2=e d$ are also positive constants.

\qed

\subsection*{Supplementary Note 2: Proof of Lemma~2}


To prove the first result, we will use the identity
\begin{align}
e^{- i s \bar H_{l}} -  e^{- i s H_{l}} = -i \int_0^s ds' e^{- i (s-s') \bar H_{l}} (\bar H_l - H_l) e^{-i s' H_l } \;.
\end{align}
We note that, from the assumptions, $\Pi_{\le \Lambda'} = \Pi_{\le \Lambda'} \Pi_{\le \Delta'}= \Pi_{\le \Delta'}\Pi_{\le \Lambda'}$
and $[\Pi_{\le \Delta'},e^{-i s \bar H_l}]=0$ for all $s$. Then, we can 
express $\Pi_{\le \Lambda'} (e^{-i s \bar H_l}- e^{-i s H_l}) \Pi_{\le \Lambda}$ as
\begin{align}
\nonumber
-i \Pi_{\le \Lambda'} \left ( \int_0^s ds' e^{- i (s-s') \bar H_l}\Pi_{\le \Delta'} (\bar H_l - H_l)   (\Pi_{\le \Delta'}+ \Pi_{>\Delta'})e^{-i s' H_l } \right) \Pi_{\le \Lambda} \;.
\end{align}
We can simplify this expression since $\Pi_{\le \Delta'} (\bar H_l - H_l) \Pi_{\le \Delta'}=0$.
Observing that $\|\Pi_{\le \Lambda'}\|=1$, $\| e^{- i (s-s') \bar H_l}\|=1$, and using standard properties of the spectral norm, we obtain
\begin{align}
\| \Pi_{\le \Lambda'} (e^{-i s \bar H_l}- e^{-i s H_l}) \Pi_{\le \Lambda} \| 
& \le \int_0^{|s|} ds' \; \| \Pi_{\le \Delta'} (\bar H_l - H_l) \Pi_{> \Delta'}\| \|  \Pi_{> \Delta'} e^{-is'H_l} \Pi_{\le \Lambda} \| \\
\label{eq:effTL1}
& \le  \| \Pi_{\le \Delta'} (\bar H_l - H_l) \Pi_{> \Delta'}\| e^{-\lambda(\Delta'-\Lambda)} \int_0^{|s|} ds' \; (e^{\alpha s' M} - 1) \\
\label{eq:effTL2}
& =  \|\Pi_{\le \Delta'} (\bar H_l - H_l) \Pi_{> \Delta'}\| \frac{ e^{-\lambda (\Delta'-\Lambda)} (e^{\alpha |s| M} - 1 - \alpha M |s|)}{\alpha M} \;,
\end{align}
where $\lambda=1/(2Jdk)$, $\alpha=eJ$, and Eq.~\eqref{eq:effTL1} follows from Lemma~1.
Note that
\begin{align}
\|\Pi_{\le \Delta'} (\bar H_l - H_l) \Pi_{> \Delta'}\| & = \| \Pi_{\le \Delta'} H_l \Pi_{> \Delta'}\| \\
& \le \| H_l \| \\
& \le JM \\
& \le \alpha M \;.
\end{align}
We obtain
\begin{align}
\label{eq:lemma3.1}
\| \Pi_{\le \Lambda'} (e^{-i s \bar H_l}- e^{-i s H_l}) \Pi_{\le \Lambda} \|  & \le e^{-\lambda (\Delta' - \Lambda)} (e^{\alpha |s| M} - 1 - \alpha M |s|) \;.
\end{align}
To simplify our analysis, we will use a looser upper bound in the statement of the Lemma~2, which follows directly from Eq.~\eqref{eq:lemma3.1} and $\Delta' \ge \Lambda'$:
\begin{align}
\label{eq:lemma3.2}
\| \Pi_{\le \Lambda'} (e^{-i s \bar H_l}- e^{-i s H_l}) \Pi_{\le \Lambda} \|  & \le e^{-\lambda (\Lambda' - \Lambda)} (e^{\alpha |s| M} - 1) \;.
\end{align}

In particular, when $k$ and $d$ are constants,
we have $\alpha M \le e J d N$, and Eq.~\eqref{eq:lemma3.2} implies
\begin{align}
    \| \Pi_{\le \Lambda'} (e^{-i s \bar H_l}- e^{-i s H_l}) \Pi_{\le \Lambda} \| & \le e^{-\alpha_1 (\Lambda' - \Lambda)/J} (e^{\alpha_2 J |s| N} - 1) \;,
\end{align}
where $\alpha_1=1/(2dk)$ and $\alpha_2=e d$ are also constants.

To prove the second result, we use Lemma~1 together with Eq.~\eqref{eq:lemma3.2} and standard properties of the spectral norm, and obtain
\begin{align}
\|\Pi_{>\Lambda'} e^{-is \bar H_l} \Pi_{\le \Lambda}\|& = \|\Pi_{\le \Delta '}\Pi_{>\Lambda'} e^{-is \bar H_l} \Pi_{\le \Lambda} \| \\
& = \| \Pi_{\le \Delta '}\Pi_{>\Lambda'} (e^{-is \bar H_l} - e^{-is  H_l} +e^{-is  H_l}) \Pi_{\le \Lambda} \| \\
& \le \| \Pi_{\le \Delta '}\Pi_{>\Lambda'} (e^{-is \bar H_l} - e^{-is  H_l}) \Pi_{\le \Lambda} \| + e^{-\lambda(\Lambda'-\Lambda)}(e^{\alpha |s| M}-1) \\
& = \| (\Pi_{\le \Delta '} -\Pi_{\le \Lambda'}) (e^{-is \bar H_l} - e^{-is  H_l}) \Pi_{\le \Lambda} \| + e^{-\lambda(\Lambda'-\Lambda)}(e^{\alpha |s| M}-1) \\
\nonumber   & \le \| \Pi_{\le \Delta '} (e^{-is \bar H_l} - e^{-is  H_l}) \Pi_{\le \Lambda} \| + \| \Pi_{\le \Lambda'} (e^{-is \bar H_l} - e^{-is  H_l}) \Pi_{\le \Lambda} \|\\ 
& \quad+ e^{-\lambda(\Lambda'-\Lambda)} (e^{\alpha |s| M}-1)
\\
& \le (e^{-\lambda(\Delta'-\Lambda)}+2e^{-\lambda(\Lambda'-\Lambda)}) (e^{\alpha |s| M}-1) \\
\label{eq:lemma2eq2}
& \le 3e^{-\lambda(\Lambda'-\Lambda)} (e^{\alpha |s| N}-1) \;.
\end{align}

In particular, when $k$ and $d$ are constants,
we have $\alpha M \le e J d N$, and Eq.~\eqref{eq:lemma2eq2} implies
\begin{align}
    \|\Pi_{>\Lambda'} e^{-is \bar H_l} \Pi_{\le \Lambda}\| & \le 3e^{-\alpha_1(\Lambda'-\Lambda)/J} (e^{\alpha_2 J |s| M}-1) \;,
\end{align}
where $\alpha_1=1/(2dk)$ and $\alpha_2=e d$ are also positive constants.

\qed

\subsection*{Supplementary Note 3: Approximation errors for product formulas}

We consider generic product formulas of $q>1$ terms,
\begin{align}
W ({\bf s}) &=  e^{- i s_q  H_{l_q}}   \cdots   e^{- i s_2  H_{l_2}}  e^{- i s_1  H_{l_1}} \;, \\
\bar W ({\bf s}) &=  e^{- i s_q  \bar H_{l_q}}   \cdots   e^{- i s_2  \bar H_{l_2}}  e^{- i s_1  \bar H_{l_1}} \;,
\end{align}
where ${\bf s}=s_1,\ldots,s_q$, $s_j \in \mathbb R$,
and $1 \le l_j \le L$. 
We also define
\begin{align}
W^{\bf \Lambda}({\bf s}) & =  \Pi_{\le \Lambda_{q}} e^{- i s_q  H_{l_q}}   \cdots \Pi_{\le \Lambda_2} e^{- i s_2  H_{l_2}}\Pi_{\le \Lambda_{1}} e^{- i s_1  H_{l_1}} \;, \\
\bar W^{\bf \Lambda}({\bf s}) & =  \Pi_{\le \Lambda_{q}} e^{- i s_q  \bar H_{l_q}}   \cdots \Pi_{\le \Lambda_2} e^{- i s_2  \bar H_{l_2}}\Pi_{\le \Lambda_1} e^{- i s_1  \bar H_{l_1}} \;,
\end{align}
where ${\bf \Lambda}=(\Lambda_1,\ldots,\Lambda_q)$, and $\Lambda_j \ge 0$ for all $j$.
Using Lemmas~1 and ~2, we now prove a number of results (corollaries) on the approximation errors for these product formulas
that will be required for the proof of Thm.~1. First, we will prove that $W({\bf s})$ produces approximately the same state as
$W^\bl({\bf s})$ when the initial state is supported on the low-energy subspace and for a suitable choice of $\bl$. Next, we will show that the approximation error from replacing $W^\bl({\bf s})$ by $\bar W^\bl({\bf s})$ is of the same order as that of the first approximation
for a suitable choice of $\Delta'$ and effective Hamiltonians.
A similar result is obtained if we further replace 
$\bar W^\bl({\bf s})$ by $\bar W({\bf s})$. Combining these results we show that the state produced by $W({\bf s})$ approximates that produced by $\bar W({\bf s})$
for a suitable choice of $\Delta'$.

\begin{corollary}
	\label{cor:projectapprox}
	Let $\delta >0$, $\Delta \ge 0$, $\lambda=1/(2Jdk)$, and $\alpha=eJ$. Then, if $\bl$  satisfies $\Lambda_j - \Lambda_{j-1} \ge \frac 1 \lambda (\alpha |s_j| M + \log(q/\delta))$ and $\Lambda_0=\Delta$,
	\begin{align}
	\| (W^{\bf \Lambda}({\bf s}) - W({\bf s}))
	\Pi_{\le \Delta} \| \le \delta \;.
	\end{align}
\end{corollary}
\begin{proof}
	We use the identity
	\begin{align}
	\nonumber
	W({\bf s}) - W^{\bf \Lambda}({\bf s})& = \Pi_{> \Lambda_q}  e^{- i s_q   H_{l_q}}\Pi_{\le \Lambda_{q-1}}  \cdots \Pi_{\le \Lambda_{1}} e^{- i s_1   H_{l_1}} + \cdots   \\
	& + e^{- i s_q   H_{l_q}} \cdots \Pi_{> \Lambda_2} e^{- i s_2   H_{l_2}} \Pi_{\le \Lambda_1}e^{- i s_1   H_{l_1}}+ e^{- i s_q   H_{l_q}} \cdots e^{- i s_2   H_{l_2}} \Pi_{> \Lambda_1}e^{- i s_1   H_{l_1}}  \;.
	\end{align}
	Since $\| e^{-i s_j   H_{l_j}}\| = \| \Pi_{\le \Lambda_j}\|=1$, we can use the triangle inequality and Lemma~1 to obtain
	\begin{align}
	\|(   W^{\bf \Lambda}({\bf s})-   W({\bf s}) ) \Pi_{\le \Delta}  \|&  \le \sum_{j=1}^{q} \| \Pi_{>\Lambda _j}
	e^{-i s_j  H_{l_j}} \Pi_{\le \Lambda_{j-1}} \| \\
	&\le \sum_{j=1}^{q} e^{-\lambda (\Lambda_j - \Lambda_{j-1})}(e^{\alpha |s_j| M }- 1 ) \\
	& \le \sum_{j=1}^{q} \delta/q \\
	& = \delta \;.
	\end{align}
\end{proof}



\begin{corollary}
	\label{cor:projectapprox2}
	Let $\delta >0$, $\Delta \ge 0$, $\lambda=1/(2Jdk)$, and $\alpha=eJ$. Then, if $\bl$  satisfies $\Lambda_j - \Lambda_{j-1} \ge \frac 1 \lambda (\alpha |s_j| M + \log(q/\delta))$, $\Lambda_0=\Delta$, and $\Delta' \ge \Lambda_q$,
	\begin{align}
	\| (\bar W^{\bf \Lambda}({\bf s}) - W^{\bf \Lambda}({\bf s}))
	\Pi_{\le \Delta} \| \le \delta \;.
	\end{align}
\end{corollary}

\begin{proof}
	We use the identity
	\begin{align}
	\nonumber
	\bar W^{\bf \Lambda}({\bf s}) -   W^{\bf \Lambda}({\bf s}) =&
	\Pi_{\le \Lambda_{q}} (e^{- i s_q  \bar H_{l_q}} -
	e^{- i s_q  H_{l_q}}) \Pi_{\le \Lambda_{q-1}}\cdots \Pi_{\le \Lambda_1} e^{- i s_1  \bar H_{l_1}} + \ldots  \\
	&+ \Pi_{\le \Lambda_{q}} e^{- i s_q  H_{l_q}} \Pi_{\le \Lambda_{q-1}}   \cdots   \Pi_{\le \Lambda_1} (e^{- i s_1  \bar H_{l_1}}- e^{- i s_1  H_{l_1}}) \;.
	\end{align}
	Since $\| e^{-i s_j   H_{l_j}}\| =\| e^{-i s_j   \bar H_{l_j}}\|=\| \Pi_{\le \Lambda_j}\| =1$, we can use the triangle inequality and Eq.~\eqref{eq:lemma3.2} in Lemma~2 to obtain
	\begin{align}
	\|(\bar W^{\bf \Lambda}({\bf s}) -   W^{\bf \Lambda}({\bf s}) )\Pi_{\le \Delta}  \| & \le \sum_{j=1}^q \| \Pi_{\le \Lambda_j} (e^{- i s_j  \bar H_{l_j}}- e^{- i s_j  H_{l_j}})\Pi_{\le \Lambda_{j-1}}\| \\
	& \le \sum_{j=1}^q e^{-\lambda (\Lambda_j - \Lambda_{j-1})} (e^{\alpha |s_j| M} - 1) \\
	& \le \sum_{j=1}^q \delta/q \\
	\label{eq:lem4.main2}
	& = \delta \;.
	\end{align}
\end{proof}

\begin{corollary}
	\label{cor:projectapprox3}
	Let $\delta >0$, $\Delta \ge 0$, $\lambda=1/(2Jdk)$, and $\alpha=eJ$. Then, if $\bl$  satisfies $\Lambda_j - \Lambda_{j-1} \ge \frac 1 \lambda (\alpha |s_j|M + \log(q/\delta))$, $\Lambda_0=\Delta$, and $\Delta' \ge \Lambda_q$,
	\begin{align}
	\| (\bar W({\bf s}) - \bar W^{\bf \Lambda}({\bf s}))
	\Pi_{\le \Delta} \| \le  3 \delta \;.
	\end{align}
\end{corollary}

\begin{proof}
	We use the identity
	\begin{align}
	\nonumber
	\bar W({\bf s}) - \bar W^{\bf \Lambda}({\bf s})& = \Pi_{> \Lambda_q}  e^{- i s_q   \bar H_{l_q}} \Pi_{\le \Lambda_{q-1}} \cdots   \Pi_{\le \Lambda_{1}} e^{- i s_1   \bar H_{l_1}} + \cdots   \\
	& + e^{- i s_q   \bar H_{l_q}} \cdots \Pi_{> \Lambda_2} e^{- i s_2   \bar H_{l_2}} \Pi_{\le \Lambda_1}e^{- i s_1   \bar H_{l_1}}+ e^{- i s_q   \bar H_{l_q}} \cdots e^{- i s_2   \bar H_{l_2}} \Pi_{> \Lambda_1}e^{- i s_1   \bar H_{l_1}}  \;.
	\end{align}
	Since $\| e^{-i s_j   \bar H_{l_j}}\| = \| \Pi_{\le \Lambda_j}\|=1$, we can use the triangle inequality  and Eq.~\eqref{eq:lemma2eq2} in Lemma~2
	to obtain
	\begin{align}
	\| ( \bar W ({\bf s}) - \bar W^{{\bf \Lambda}} ({\bf s}) )\Pi_{\le \Delta}  \| & \le
	\sum_{j=1}^q \| \Pi_{>\Lambda_j} e^{-is_j \bar H_{l_j}} \Pi_{\le \Lambda_{j-1}} \| \\
	& \le  \sum^{q}_{j=1} 3 e^{-\lambda (\Lambda_j - \Lambda_{j-1})} (e^{\alpha |s_j| M} - 1) \\
	& \le 3 \sum^{q}_{j=1} \delta/q \\
	& = 3 \delta \label{eq:lem4.main3} \;.
	\end{align}
\end{proof}

\begin{corollary}
	\label{cor:projectapprox4}
	Let $\delta >0$, $\Delta \ge 0$, $\lambda=1/(2Jdk)$, $\alpha=eJ$, and $|{\bf s}|=\sum_{j=1}^q |s_j|$. Then, if $\Delta'\ge \Delta +  \frac{1}{\lambda} \left(\alpha |{\bf s}| M + q \log(q/\delta)\right)$,
	\begin{align}
	\| (\bar W({\bf s}) - W({\bf s}))
	\Pi_{\le \Delta} \| \le  5 \delta \;.
	\end{align}
\end{corollary}

\begin{proof}
	We define the energies $\Lambda_q \ge \ldots \ge \Lambda_0=\Delta$ via
	\begin{align}
	\Lambda_j - \Lambda_{j-1} = \frac 1 \lambda (\alpha |s_j|M + \log(q/\delta)) \;.
	\end{align}
	In particular, $\Delta'\ge \Lambda_q = \Delta +  \frac 1 \lambda \left(\alpha |{\bf{s}}| M + q \log(q/\delta)\right)$. 
	We use the identity
	\begin{align}
	\bar W({\bf s}) - W({\bf s}) = (\bar W({\bf s}) -\bar W^{{\bf \Lambda}}({\bf s})) +(\bar W^{{\bf \Lambda}}({\bf s}) - W^{{\bf \Lambda}}({\bf s})) +(W^{{\bf \Lambda}}({\bf s}))
	- W({\bf s})) \;.
	\end{align}
	The triangle inequality and Corollaries~\ref{cor:projectapprox3},~\ref{cor:projectapprox2}, and~\ref{cor:projectapprox} imply
	\begin{align}
	\nonumber
	\|   ( \bar W({\bf s}) - & W({\bf s}) )\Pi_{\le \Delta} \| & \\ 
	& \le \| (\bar W({\bf s}) -\bar W^{{\bf \Lambda}}({\bf s})) \Pi_{\le \Delta}\|+ \|(\bar W^{{\bf \Lambda}}({\bf s}) - W^{{\bf \Lambda}}({\bf s})) \Pi_{\le \Delta}\| + \| (W^{{\bf \Lambda}}({\bf s}))
	- W({\bf s})) \Pi_{\le \Delta} \| \\
	& \le 3 \delta + \delta + \delta \\
	& = 5 \delta \;.
	\end{align}
\end{proof}

To prove the previous corollaries,
we often used $(e^{\alpha |s_j| M }-1) \le e^{\alpha |s_j| M } $, but we note that a different upper bound such as $(e^{\alpha |s_j| M }-1) \le \alpha |s_j| M e^{\alpha |s_j| M }$ may provide improved results (e.g., a smaller value of $\Delta'$ in Cor.~\ref{cor:projectapprox4}), especially if  $\alpha |s_j| M \ll 1$. However, while this other upper bound could be used to improve the constants hidden by the $\cO$ notation in the complexity of our method, we were unable to improve its asymptotic scaling from using $(e^{\alpha |s_j| M }-1) \le \alpha |s_j| M e^{\alpha |s_j| M }$. 


\subsection*{Proof of Thm.~1}

For some $\Delta' \ge \Delta \ge 0$, let $\bar U(s) = e^{-i s \bar H}$ be the evolution operator with the effective Hamiltonian and $\bar W_p(s)$,
$p \ge 1$,
be the corresponding $p$-th order product formula obtained by replacing
$H_l \rightarrow \bar H_l$ in $W_p(s)$ of Eq.~1. Since the condition $H_l \ge 0$ implies $\| \bar H_l \| \le \Delta'$,
we obtain
\begin{align}
\| (\bar U(s) - \bar W_p(s)) \Pi_{\le \Delta} \| & \le \|  \bar U(s) - \bar W_p(s) \| \\
\label{eq:SM:trottererror}
& \le \epsilon(\Delta') \;,
\end{align}
where $\epsilon (\Delta')= \gamma (L \Delta' |s|)^{p+1}$ is an upper bound of the error induced by product formulas using effective operators and $\gamma=\cO(1)$ is a constant~\cite{CST+19}. 
This error bound grows with $\Delta'$. It does not exploit any structure of the effective Hamiltonians so it may be possible to improve it under further constraints.

Additionally,
\begin{align}
\|  (U(s) - \bar U(s)) \Pi_{\le \Delta} \| = 0\;, 
\end{align}
and then
\begin{align}
\| (U(s) - \bar W_p(s)) \Pi_{\le \Delta} \| & \le \epsilon (\Delta')\;.
\end{align}
The other contribution to the error is due to Cor.~\ref{cor:projectapprox4}, which can be turned around to obtain a bound on the error that depends on $\Delta'$.
Let $\lambda = 1/(2Jdk)$, $\alpha =eJ$, and $q>1$ be the number of terms in the product $\bar W_p(s)$. Then, Cor.~\ref{cor:projectapprox4} implies
\begin{align}
\| (\bar W_p(s) - W_p(s))
\Pi_{\le \Delta} \| \le  5 \delta (\Delta')\;,
\end{align}
and
\begin{align}
\delta (\Delta')=  e^{-\frac{1}{q} (\lambda (\Delta' - \Delta) -  \alpha |{\bf s}| M -q \log q)}  \;.
\end{align}
This error bound decreases when $\Delta'$ increases. It is now valid for all $\Delta' \ge 0$ but can be irrelevant (larger than 1) when, for example, $\Delta' \le \Delta$. 
We assume that our product formula is such that $|{\bf s}|\leq \kappa L |s|$ for a constant $\kappa \ge 1$ and let $\alpha'= \kappa \alpha$. Then 
\begin{align}
\delta (\Delta')= e^{- \frac{1}{q}(\lambda (\Delta' - \Delta) -  \alpha' |s| M L - q \log q )}  \;.
\end{align}
Thus, for any $\Delta' \ge \Delta$, the triangle inequality implies $\| (U(s) - W_p(s))\Pi_{\le \Delta}\| \le \epsilon(\Delta') + 5 \delta(\Delta')$.

For given $|s|$, we can search for $\Delta' \ge \Delta$ that minimizes the overall error bound. 
Let that $\Delta'$ be $\Delta'_{\min}$. The global minimum for the error is then $\epsilon(\Delta'_{\min}) + 5 \delta (\Delta'_{\min})$ and,  for all $\Delta' \ge \Delta$, it satisfies
\begin{align}
\epsilon(\Delta'_{\min}) + 5 \delta (\Delta'_{\min}) \le
\epsilon(\Delta') + 5 \delta (\Delta') \;.
\end{align}

Then, we can fix any value of $\Delta' \ge \Delta$ and obtain a bound for the overall error from computing $\epsilon(\Delta') + 5 \delta (\Delta')$.
In particular, we choose
\begin{align}
\label{eq:D'def}
\Delta' &=\Delta + \frac{1}{\lambda}\alpha' |s| M L +\frac{q}{\lambda}  \log q + \frac{q}{\lambda} (p+1) \log \left( \frac{1}{J |s|} \right) \;,
\end{align}
and, assuming $J |s| \le 1$, $q>1$, we obtain
\begin{align}
\delta(\Delta') & = (J|s|)^{p+1}\\ 
& \le (\Delta' |s|)^{p+1} \\
& \le  (L \Delta' |s|)^{p+1} \;.
\end{align}
The constraint in $J|s|$ is to avoid errors larger than 1: it is sufficient for $\delta(\Delta')\le 1$ and for $\Delta' \ge \Delta$.
Therefore, if $J|s| \le 1$,
\begin{align}
\|  (U(s) - W_p(s)) \Pi_{\le \Delta} \| & \le \epsilon(\Delta'_{\min}) + 5 \delta (\Delta'_{\min}) \\
& \le \epsilon(\Delta') + 5 \delta(\Delta') \\
& \le (\gamma+5) (L \Delta' |s|)^{p+1} \\
& \le \tilde{\gamma} (L  \Delta' |s|)^{p+1} \\
\label{eq:tildeepsilon}
& = \tilde \epsilon(\Delta')\;,
\end{align}
where $\Delta'$ was determined in Eq.~\eqref{eq:D'def} and
$\tilde \gamma= \gamma + 5$ is a constant.
While our choice of $\Delta'$ in Eq.~\eqref{eq:D'def} 
does not correspond to the global  minimum
of the overall error, an exact calculation show that it is not far from $\Delta'_{\min}$ and 
provides the same asymptotic complexity for our method. Nevertheless, if one is interested in optimizing some constants hidden by the $\cO$ notation in the overall complexity, using $\Delta'_{\min}$ instead will be a better choice. Also note that, in general, we can additionally require $\Delta ' \le \|H\|$ for all $s$; this condition is not satisfied by Eq.~\eqref{eq:D'def}. But since our results will be particularly useful when $\Delta' \ll \|H\|$, we do not expect any improvement from this additional requirement and only consider Eq.~\eqref{eq:D'def}
in the following.

An upper bound for $\Delta'$ can be given 
in terms of three factors $\beta_1$, $\beta_2$, and $\beta_3$, which can be easily computed from the parameters that define $H$ as $\Delta'=\Delta + \beta_1 J \log(\beta_2/(J|s|))+ \beta_3 J^2 N |s|$.
This is the expression provided in Thm.~1 in the main text.
Using the properties $q>1$, $p \ge 1$, and $N \le ML \le dN$, the factors satisfy 
\begin{align}
\beta_1 & = 2 q dk (p+1)\\
\beta_2 & = q^{1/(p+1)} \;,\\
\beta_3 & \le 2 e k d^2 \kappa \;.
\end{align}
When $d$,  $k$, $L$, and $q$ are $\cO(1)$ constants, 
we obtain $\beta_1=\cO(1)$, $\beta_2=\cO(1)$ and $\beta_3=\cO(1)$.

\qed

\vspace{0.5cm}

Note that the error $\tilde \epsilon(\Delta')$
approaches zero as $s \rightarrow 0$ but
not as $|s|^{p+1}$, as in the case of $p$-th order product formulas. 
This is due to the appearance of $\log(1/(J|s|))$ in $\Delta'$. This logarithmic factor is the result of our  requirement that the ``leakage'' $\delta(\Delta')$ vanishes as $s \rightarrow 0$, implying $\Delta' \rightarrow \infty$. 
At the same time, when $s \rightarrow 0$, no evolution occurs
and one would have expected that $\Delta' \rightarrow \Delta$
in this limit. This suggests that it might be possible to obtain a tighter expression for $\Delta'$ than that in Eq.~\eqref{eq:D'def}, one that does not diverge in the limit $s \rightarrow 0$, but leave this as an open problem.
Nevertheless, our purpose is to make $|s|$ as large as possible in the product formula and the limit $s \rightarrow 0$ is not of our concern.
To this end, our previous analyses will suffice.


For the special case of Trotter-Suzuki product formulas,
the constant $\gamma$ that appears in $\epsilon(\Delta')$
has been previously studied in Ref.~\cite{BAC07}. In this case, the number of terms satisfies $q \le 5^p L$ and $|{\bf s}| \le c^p L |s|$, for some constant $c \approx 2.32$. Thus, we can take $\kappa=c^p$ and $\kappa=\cO(1)$ for $p=\cO(1)$.

The assumption $H_l \ge 0$ implies $\|\bar H_l\| \le \Delta'$
and thus the error bound in Eq.~\eqref{eq:SM:trottererror}.
In general, $H_l \ge 0$ can be met after a simple shifting $H_l \rightarrow H_l + a_l$, and the assumption seems irrelevant.
However, this shifting could result on a value of $\Delta$ and $\Delta'$ that scales, for example, with $N$ or $\|H_l\|$. If this is the case, the error bound of Eq.~\eqref{eq:SM:trottererror} could be comparable to that in the worst case (without the low-energy assumption), and would not result in a complexity improvement.

To clarify this point further, consider a local spin Hamiltonian where $H_l \ngeq 0$ and let $\Delta$ be the relevant low-energy parameter; that is, the initial state is supported on the subspace associated with eigenvalues of $H$ in the range $[E_0,E_0+\Delta]$, where $E_0$ is the lowest eigenvalue. 
The previous analyses can be extended to this general case as follows.
First, we let $\Pi_{\le \Delta'}$ be the projector into the subspace associated with the eigenvalues of $H$ in the range $[E_0,E_0+\Delta']$, where $\Delta' > \Delta$ is chosen as in Thm.~1. (As before, the effective Hamiltonians are $\bar H_l = \Pi_{\le \Delta'} H_l \Pi_{\le \Delta'}$. ) We also define
\begin{align}
    \tilde \Delta& := \max_l \min_a \|\bar H_l - a\| \\
    & = \max_l \frac 1 2 |E_l^{\max} - E_l^{\min}|\;,
\end{align}
which is an effective low-energy norm (obtained by a minimization that eliminates a constant term in $\bar H_l$), and $E_l^{\max}$ and $E_l^{\min}$ are the largest and smallest eigenvalues of $\bar H_l$, respectively. Then, the error bound induced by product formulas in the more general case where $H_l \ngeq 0$ is
\begin{align}
    \|\bar U(s) - \bar W_p(s) \| \le \epsilon(\tilde \Delta) \;,
\end{align}
where $\epsilon(\tilde \Delta)=\gamma(L \tilde \Delta |s|)^{p+1}$. This is because this error bound depends on nested commutators of the $\bar H_l$'s~\cite{CST+19}, and shifting $\bar H_l$ by a constant does not change the commutators.
This error has to be compared with Eq.~\eqref{eq:SM:trottererror}. A complexity improvement with respect to the worst case (without the low-energy assumption) might follow if $\tilde \Delta \ll \|H_l\|$; e.g., when $\tilde \Delta$ is constant. 

Characterizing those Hamiltonians for which $\tilde \Delta \ll \|H_l\|\le h$ occurs is an interesting but open problem. Nevertheless, for many important local Hamiltonians (e.g., frustration-free Hamiltonians), the condition $H_l \ge 0$ is readily satisfied. In these examples, $\Delta'$ (or $\tilde \Delta$) can be much less than $\|H_l\|$, resulting in a complexity improvement. The Heisenberg model of Fig. 1 is an example where our results apply.

\subsection*{Supplementary Note 4: Complexity of product formulas for local Hamiltonians}

We now determine the complexity of product formulas, which is the total number of exponentials of the $H_l$'s
to approximate the evolution operator $U(t)$ within given
precision $\varepsilon>0$. We give an explicit dependence of this complexity
in terms of the relevant parameters that specify $H$.
If a quantum algorithm is constructed to implement such product formula, then our result will determine the complexity of the quantum algorithm (number of two-qubit gates) from multiplying it by the complexity of implementing the exponential of $H_l$. The latter is linear in $kM$ following the results in Ref.~\cite{SOGKL02}.

\begin{theorem}
	\label{thm:complexity}
	Let $\varepsilon>0$, $\Delta \ge 0$, $t \in \mathbb R$, $H=\sum_{l=1}^L H_l$ a $k$-local Hamiltonian as above, $H_l \ge 0$, and $W_p(s)$ a $p$-th order product formula as in Eq.~1. Then, there exists
	\begin{align}
	r 
	&= \tilde{\mathcal{O}} \left(\frac{t^{1+ \frac 1 p}}{\varepsilon^ {\frac 1 p}} \left( L \Delta +  L d k q (\log q)J \right)^{1+ \frac 1 p} \right) + \cO\left( \frac{t^{1+ \frac 1 {2p+1}}}{\varepsilon^{ \frac 1 {2p+1}}}(L^2 d M  J^2)^{\frac 1 2 + \frac  1 {4p+2}}
	\right) \;,
	\end{align}
	such that
	\begin{align}
	\| (U(t) - (W_p(t/r))^r ) \Pi_{\le \Delta} \| \le \varepsilon \; .
	\end{align}
	The $\tilde \cO$ notation hides a polylogarithmic factor in
	$(|t| J L q d k/\varepsilon)$.
\end{theorem}

\begin{proof}
	Let $r=t/s$ be the Trotter number, i.e., the number of ``segments'' in the product formula, each approximating the evolution $U(s)$ for short time $s$. We assume $s \ge 0$ and $t\ge 0$ for simplicity, and the result  for $t\le 0$ simply follows from replacing $t \rightarrow |t|$.
	Note that $U(s) \Pi_{\le \Delta}=\Pi_{\le \Delta}
	U(s) \Pi_{\le \Delta}$ and $U(t)=(U(s))^r$.
We use the identity
	\begin{align}
	(U(t)&-(W_p(s))^r )\Pi_{\le \Delta}  = \sum^{r-1}_{r'=0} (W_p(s))^{r'} (U(s)-W_p(s)) (U(s))^{r-r'-1} \Pi_{\le \Delta}\\
	&= \sum^{r-1}_{r'=0} (W_p(s))^{r'} (U(s)-W_p(s) ) \Pi_{\le \Delta} (U(s))^{r-r'-1} \Pi_{\le \Delta} \;.
\end{align}

Using the triangle inequality and $\|W_p(s)\|= \|U(s)\|=1$ for unitary operators, we obtain

	\begin{align}
	\| (U(t) - (W_p(s))^r)\Pi_{\le \Delta} \| & \le r  \|(U(s)-W_p(s)) \Pi_{\le \Delta} \|\\
	& \le r  \tilde \epsilon(\Delta')\;,
	\end{align}
	where $\tilde \epsilon(\Delta')= \tilde{\gamma} (L\Delta' s)^{p+1}$ and $\tilde \gamma$ is a constant that can be determined from the error bounds of product formulas -- see Thm.~1, Eq.~\eqref{eq:tildeepsilon}.
	Thus, for overall error bounded by $\varepsilon$,
	it suffices to choose $\tilde \epsilon(\Delta') \le \varepsilon/r$ or, equivalently,
	\begin{align}
	\label{eq:scond0}
	\tilde{\gamma} (L\Delta'  s)^{p+1} \le \frac \varepsilon {t} s \;
	\end{align}

    To set some conditions in $s$, in addition to $Js \le 1$,
	we note that each of the terms in Eq.~\eqref{eq:D'def}
	that define the effective norm
	can be dominant depending on $s$, $\Delta$, and other parameters. Then, to obtain the overall complexity of product formulas, we will analyze four different cases as follows.
	In the first case, we assume that $\Delta$ is the dominant term in $\Delta'$, and we require
	\begin{align}
	\tilde{\gamma} (4L  \Delta s)^{p+1} \le \frac \varepsilon {t} s \;,
	\end{align}
	which is satisfied as long as 
	\begin{align}
	\label{eq:scond1}
	s \le s_1=\left(\frac \varepsilon {\tilde \gamma t } \right)^{1/p}
	\frac 1 { ( 4  L \Delta )^{1+1/p} } \;.
	\end{align}

	In the second case, we assume that $\alpha's M L/\lambda$ is the dominant term in Eq.~\eqref{eq:D'def} and impose
	\begin{align}
	\tilde{\gamma} \left(4L \frac{ \alpha' M L s^2}{\lambda} \right)^{p+1} \leq \frac{\varepsilon}{t} s,
	\end{align}
	which implies
	\begin{align}\label{eq:scond2}
	s \leq s_2= \left( \frac{\varepsilon}{\tilde \gamma t} \right)^{\frac{1}{2p+1}} \left( \frac{\lambda}{4   \alpha' M L^2} \right)^{\frac{1}{2} + \frac{1}{4p+2}}. 
	\end{align}
	
	In the third case, we assume that $(q\log q)/\lambda$ is the dominant term in Eq.~\eqref{eq:D'def} and impose
	\begin{align}
	\tilde{\gamma} \left(4L \frac{ q \log q \; s}{\lambda} \right)^{p+1} \leq \frac{\varepsilon}{t} s,
	\end{align}
	which implies
	\begin{align}\label{eq:scond3}
	s \leq s_3= \left( \frac{\varepsilon}{\tilde{\gamma} t} \right)^{\frac{1}{p}} \left( \frac{\lambda}{4 L q \log q} \right)^{1 + 1/p}. 
	\end{align}

	In the fourth case, we assume that $\frac{q}{\lambda} (p+1) \log\left(\frac{1}{Js}\right)$ is the dominant term in Eq.~\eqref{eq:D'def} and impose
	\begin{align}
	\tilde{\gamma} \left(4 L \frac{q }{\lambda} (p+1) \log\left(\frac{1}{J s}\right) s \right)^{p+1} \leq \frac{\varepsilon}{t} s \;,
	\end{align}
	under the assumption $J s \le 1$.
	Equivalently, if $z= Js \leq 1$ and 
	defining the function $f(z)=(\log(1/z))^{p+1} z^{p}$, $p \ge 1$, we impose
	\begin{align}
	\label{eq:scond4a}
	f(z) \leq X \;, 
	\end{align}
	where 
	\begin{align}
	X= \frac{\varepsilon}{ \tilde{\gamma} t }  \left( \frac{\lambda}{4 L q(p+1) } \right)^{p+1} J^p.
	\end{align}
	To set a fourth condition in $s$ we could then compute $X$
	from the inputs of the problem and find a range of values of $z$ for which Eq.~\eqref{eq:scond4a} is satisfied. We can also obtain such a range analytically as follows.
	The function $f(z)$ increases from $f(0)=0$, attains its maximum at $z_M= e^{-\frac{p+1}{p}}$ (hence $e^{-2} \le z_M \le e^{-1}$
	for all $p \ge 1$), and then decreases to $f(1)=0$.
	Additionally, $f(z) \le ((1+1/p)/e)^{p+1} \le 4/e^2 \approx 0.54$ for all $0 \le z \le 1$.  In particular, 
	if $X \ge ((1+1/p)/e)^{p+1}$ then Eq.~\eqref{eq:scond4a} is readily satisfied for all $z \le 1$ and no additional condition in $s$
	will be required in this case (this happens, for example, for sufficiently small values of $t$). More generally, for a given $X$,
	we can solve for $f(z)=X$. If there are two solutions $z_{1,2}\le 1$ for $z$,
	we consider the smaller one ($z_1 < z_2$) and the relevant range for $z$ to satisfy Eq.~\eqref{eq:scond4a} 
	is $[0,z_1]$. It will then suffice to impose that $z$ belongs to a range $[0,z'_1]$, where $z'_1 \le z_1$. To this end,
	we define $z_1'=X^{1/p}/(e^2\log(e^2/X))^{(p+1)/p}$ and,
	under the assumption $X \lesssim 0.54$, we have $z_1'< e^{-2}$.
	Additionally,
	\begin{align}
	f(z) & = (\log(1/z))^{p+1} z^{p} \\
	& \le \left(\log \left(\frac{e^2 \log^{(p+1)/p}(e^2/X)}{X^{1/p}} \right) \right)^{p+1} \frac{X}{(e^2 \log(e^2/X))^{p+1}}\\
	& \le \left(\log \left(\frac{e^2 \log^{2}(e^2/X)}{X} \right) \right)^{p+1} \frac{X}{(e^2 \log(e^2/X))^{p+1}}\\
	\\
	& \le \left(3 \log (e^2/X) \right)^{p+1} \frac{X}{(e^2 \log(e^2/X))^{p+1}}\\
	& \le  X \;,
	\end{align}
	where we used $X^{1/p} \ge X$ and $\log(e^2/X) \le e^2/X$ for $X \le 1$.
	This is the condition of Eq.~\eqref{eq:scond4a}.
	Then, for the fourth condition in $s$, we impose $z \le z_1'$ or, equivalently,
	\begin{align}
	s\leq s_4&=
	\left(\frac{\varepsilon}{\tilde{\gamma} t}\right)^{1/p} \left( \frac{\lambda}{4L q (p+1)} \right)^{1+1/p} \frac 1 {\left( e^2 \log \left( \frac{e^2 \tilde \gamma t} {\varepsilon J^p} \left( \frac{4Lq (p+1)} {\lambda}\right)^{p+1} \right) \right)^{1+1/p}} \\
	&=
	\label{eq:scond4}
	\left(\frac{\varepsilon}{\tilde{\gamma} t}\right)^{1/p} \left( \frac{\lambda}{4L q (p+1)} \right)^{1+1/p} \frac 1 {\left( e^2 \log \left( \frac{e^2 \tilde \gamma t J} {\varepsilon} (8Lq (p+1) dk)^{p+1} \right) \right)^{1+1/p}}
	\; .
	\end{align}
	Except for a mild polylogarithmic correction in $t J L q d k /\varepsilon$ -- the third factor -- this condition is similar to the first and third ones.
	\\
	
	Then, if $ J s \le 1$ and $s$ additionally satisfies Eqs.~\eqref{eq:scond1},~\eqref{eq:scond2},~\eqref{eq:scond3}, and~\eqref{eq:scond4}, we obtain the desired condition of Eq.~\eqref{eq:scond0}. The product formulas under consideration [Eq.~1] are such that $\alpha'=\kappa \alpha = \cO(J)$ and $\tilde \gamma=\cO(1)$, and we consider the case where $p$ is a $\cO(1)$ constant. The conditions in $s$ allow us to obtain 
	a sufficient condition for the Trotter number as follows:
	\begin{align}
	&r  = t/s \\
	\label{eq:rbound1}
	&= \tilde{\mathcal{O}} \left(\frac{t^{1+ \frac 1 p}}{\varepsilon^ {\frac 1 p}} \left( L \Delta +  L d k q (\log q)J \right)^{1+ \frac 1 p} \right) + \cO\left( \frac{t^{1+ \frac 1 {2p+1}}}{\varepsilon^{ \frac 1 {2p+1}}}(L^2 d M  J^2)^{\frac 1 2 + \frac  1 {4p+2}}
	\right) \;,
	\end{align}
	where the $\tilde \cO$ notation hides a polylogarithmic
	factor in $tJLqdk/\varepsilon$ coming from Eq.~\eqref{eq:scond4}. 
	For the case when $q$ is $\cO(L)$, $k=\cO(1)$, $d= \cO(1)$, and hence $L= \cO(1)$ and $M=\cO(N)$, and considering the asymptotic limit, we obtain
	\begin{align}
	\label{eq:rboundMainTextapp}
	r= \tilde{\mathcal{O}} \left(\frac{(t(\Delta+J))^{1+ \frac 1 p}}{\varepsilon^{\frac 1 p}}\right)  + \cO \left( \frac{(t J \sqrt N)^{1+\frac 1 {2p+1}}}{\varepsilon^{\frac 1 {2p+1}}}\right) \;.
	\end{align}
	
\end{proof}

\subsection*{Supplementary Note 5: Comparison with known results on product formulas}

We compare our result on the complexity of product formulas with those in Ref.~\cite{CST+19} 
that are the state of the art.  When no assumption is made
for the initial state, the Trotter number stated in
Ref.~\cite{CST+19} for $k$-local Hamiltonians is
\begin{align}
\tilde{r}= \cO\left( \frac{t^{1+1/p}}{\varepsilon^{1/p}} \|H\|_{ind-1} \|H\|^{1/p}_1 \right) \;.
\end{align}
Here, $\|H\|_1$ is the 1-norm of $H$, given by $\sum_{l=1}^L \|H_l\|$ in our case, and 
$\|H\|_{ind-1}$ is the so-called induced 1-norm of $H$.
The latter is defined as follows. We write $H=\sum_{j_1,\ldots,j_k=1}^N h_{j_1,\ldots,j_k}$, where each $h_{j_1,\ldots,j_k}$ includes the $k$-local interaction terms of qubits labeled as $j_1,\ldots,j_k$ in $H$. Then,
\begin{align}
\|H\|_{ind-1}:= \max_l \max_{j_l} \sum_{j_1, \ldots, j_{l-1}, j_{l+1}, \ldots, j_k=1}^N \|h_{j_1, \ldots, j_k}\| \;.
\end{align}
That is, we fix certain qubit $j_l$ and consider all the interaction terms that contain that qubit.
For a degree $d$ Hamiltonian with $k$-local interaction terms (not necessarily geometrically local), each of strength at most $J$, $\|H\|_{ind-1} \leq dJ$. Furthermore, $\|H\|_1$ can be upper bounded as 
$\|H\|_1 \leq JML \le JdN$.
As a result, for a $k$-local Hamiltonian as above, the best known upper bound for the Trotter number is
\begin{align}\label{eq:rboundUsingLocality}
\tilde{r}= \cO\left( \frac{t^{1+1/p}}{\varepsilon^{1/p}} (dJ)^{1+1/p} N^{1/p} \right).
\end{align}

To compare our main result with Eq.~\eqref{eq:rboundUsingLocality}, we express Eq.~\eqref{eq:rbound1} in terms of $d$ and $N$. We recall that the total number of local terms in $H$ is $ML \le dN$, $L \le dk+1$, and we assume that $k=\cO(1)$, $M=\cO(N)$, and $q= \mathcal{O}(L)=\cO(d)$
for the $p$-th order product formula. Then, in the asymptotic limit,
\begin{align}
\label{eq:rbound2}
r
&= \tilde{\mathcal{O}} \left(\frac{t^{1+ \frac 1 p}}{\varepsilon^ {\frac 1 p}} \left( d \Delta +  d^3 J \right)^{1+ \frac 1 p}\right) + \cO \left(\frac{t^{1+ \frac 1 {2p+1}}}{\varepsilon^{ \frac 1 {2p+1}}}(d^3 N  J^2)^{\frac 1 2 + \frac  1 {4p+2}}
\right) \;.
\end{align}
The results for various values of $p$ and $k=\cO(1)$ are in Table~\ref{fig:Comparison2}.  

\vspace{0.2 cm}

\begin{table}[htb]
	\begin{tabular}{ |p{1cm}||p{4cm}|p{9cm}| }
		\hline
		\multicolumn{3}{|c|}{Comparison of asymptotic complexities (Trotter number) as a function of $\Delta, J, d, N, \varepsilon, t$} \\
		\hline
		Order & Previous result ($\tilde{r}$) & Low-energy simulation ($r$)\\
		\hline
		$p=1$  & $\cO(t \frac{t}{\varepsilon}  d^2  N J^2)$    & $ \tilde \cO(t\frac{t}{\varepsilon} (d\Delta +d^3 J)^2 ) + \cO( t \left( \frac{t}{\varepsilon} \right)^{1/3}  d^2 N^{2/3} J^{4/3})$ \\
		\hline
		$p=2$  & $\cO(t \left(\frac{t}{\varepsilon}\right)^{1/2} d^{3/2}  N^{1/2} J^{3/2} $  )  & $ \tilde \cO( t\left(\frac{t}{\varepsilon}\right)^{1/2} (d \Delta + d^3 J)^{3/2} ) + \cO( t \left( \frac{t}{\varepsilon} \right)^{1/5} d^{9/5} N^{3/5} J^{6/5} )$ \\
		\hline
		$p=3$  & $\cO(t \left(\frac{t}{\varepsilon}\right)^{1/3} d^{4/3}   N^{1/3} J^{4/3}$  )  & $ \tilde \cO( t\left(\frac{t}{\varepsilon}\right)^{1/3} (d \Delta + d^3 J)^{4/3} ) + \cO( t \left( \frac{t}{\varepsilon} \right)^{1/7}  d^{12/7} N^{4/7} J^{8/7})$ \\
		\hline
	\end{tabular}
	\caption{The comparison between the best-known worst-case complexity of $p$-th order product formulas~\cite{CST+19} and our result for the low-energy simulation,  $p=1, 2, 3$.
	}
	\label{fig:Comparison2}
\end{table}

By setting a constraint on the initial state, our low-energy simulation result provides an advantage 
in certain regimes where $d$ may grow with $N$.
In the following, we will assume that, e.g., $\Delta =\cO( d^2 J)$ and fix the value of $t/\varepsilon$, to simplify the expressions. 
Under these assumptions, for $p=1$ and $d=\cO(N^{1/4})$,
the terms in both columns of Table~\ref{fig:Comparison2}
are of order $N^{3/2}$. 
For $p=2$ and $d=\mathcal{O}(N^{1/6})$,
the terms in both columns of Table~\ref{fig:Comparison2}
are of order $N^{3/4}$. 
For $p=3$ and $d = \mathcal{O}(N^{1/8})$
the terms in both columns of Table~\ref{fig:Comparison2}
are of order $N^{1/2}$. 
For general $p \ge 1$ and $d=\cO(N^{\frac{1}{2(p+1)}})$, 
both complexities are comparable and of order $N^{3/(2p)}$.

When $d$ is fixed and $N$ grows, our complexities
may be worse than those obtained in Ref.~\cite{CST+19}.
One reason for this is because the effective norms
may grow large in this case and the error bound that we use for product formulas using effective operators 
do not exploit any structure such as the locality of interactions. Better error bounds may be possible in this case, resulting in improved complexities. However, even if $N$ is large, our results regain an advantage at certain values of $t$, in particular if we scale $t/\varepsilon$
with $N$. Doing so will set a bound on the effective norm
so that the second term in our complexities stops dominating.

The special case where $k=\cO(1)$, $d=\cO(1)$, and $\Delta$ is also a constant independent of other parameters that specify $H$ can be directly obtained from Table~2 and is given in Table~1 in the main text.

\end{document}